\documentclass[11pt]{article}
\usepackage[letterpaper]{geometry}
\usepackage{amsmath,amsthm,amsfonts,amssymb}
\usepackage{enumerate,color,xcolor}
\usepackage{graphicx}
\usepackage{subfigure}
\usepackage{url}

\numberwithin{equation}{section}
\theoremstyle{plain}

\newtheorem{theorem}{Theorem}

\newtheorem{definition}[theorem]{Definition}

\newtheorem{proposition}[theorem]{Proposition}

\newtheorem{remark}[theorem]{Remark}
\newtheorem{assumption}[theorem]{Assumption}
\usepackage{comment}

\usepackage{ulem}

\usepackage{amsmath,amssymb,amsthm,dsfont}
\numberwithin{equation}{section}
\usepackage[colorlinks,linkcolor=orange,citecolor=blue]{hyperref}

\begin{document}

\begin{center}
  \Large \bf Asian options for local-stochastic volatility models in the short-maturity regime 
\end{center}

\author{}
\begin{center}
{Dan Pirjol}\,\footnote{School of Business, Stevens Institute of Technology, United States of America;
  dpirjol@gmail.com},
  Lingjiong Zhu\,\footnote{Department of Mathematics, Florida State University, United States of America; zhu@math.fsu.edu
 }
\end{center}

\begin{center}
 \today
\end{center}

\begin{abstract}
We derive the short-maturity asymptotics for Asian option prices in local-stochastic volatility (LSV) models. 
Both out-of-the-money (OTM) and at-the-money (ATM) asymptotics are considered. 
Using large deviations theory methods, the asymptotics for the OTM options are expressed as 
a rate function which is represented as a two-dimensional variational problem. 
We develop a novel expansion method for the variational problem by expanding the rate function 
around the ATM point.
In particular, we derive series expansions in log-moneyness for the solution of this variational problem around the ATM point,  
and obtain explicit results for the first three terms. 
We give the ATM volatility level, skew and convexity of the implied volatility of 
an Asian option in a general local-stochastic volatility model, which can be used as an approximation for pricing Asian options with strikes sufficiently close to the ATM point.
Using numerical simulations in the SABR, Heston and an LSV model with bounded
local volatility, we show good performance 
of the asymptotic result for Asian options with sufficiently small maturity. 
\end{abstract}

%%%%%%%%%%%%%%%%%%%%%%%%%%
\section{Introduction}

Asian options are popular instruments traded in many financial markets,
on underlyings such as commodity futures, equities, indices and currency
exchange rates. Compared to vanilla European options, they have the
advantage that they are less sensitive to short term price fluctuations
of the underlying asset, due to their averaging property. Typically, such
options have payoffs of the form:
\begin{equation}\label{payoff}
\mathrm{Call}\,\,\,\mathrm{Payoff} = \left(\frac{1}{T}\int_{0}^{T}S_{t}dt - K\right)^{+}\,,
\qquad
\mathrm{Put}\,\,\,\mathrm{Payoff} = \left(K-\frac{1}{T}\int_{0}^{T}S_{t}dt\right)^{+}\,,
\end{equation}
for the call and put options where $S_{t}$ is the asset price at time $t$, $K>0$ is the strike
price, $T>0$ is the maturity, and $x^{+}$ denotes $\max(x,0)$ for any $x\in\mathbb{R}$. 
Although in practice the averaging is in discrete time (daily averaging), 
it is convenient to use the continuous-time average as in \eqref{payoff}.

There is a vast literature on pricing Asian options under a wide variety of models.
The most studied case is that of the Black-Scholes model, where we mention 
several analytical approaches: the Geman and Yor method \cite{GY,DufresneReview}, 
the spectral expansion method \cite{Linetsky2004} and the Laguerre expansion method \cite{Dufresne2000}.
PDE methods can be used for pricing Asian options in a wide
variety of models \cite{RogersShi, Vecer}. The PDE
approach has been used to obtain analytical expansions in the short-maturity limit for Asian option prices under the local volatility model \cite{FPP2013}.
A diffusion operator integral expansion has been proposed in \cite{Ding2023}. 
Asian option pricing in jump-diffusion models has also been studied by Bayraktar and Xing \cite{BayraktarXing}.

Cai and Kou (2012) \cite{CaiKou} generalized the approach of \cite{GY} to a general class of jump diffusion models. This approach requires inverting a double Laplace transform.
Cai, Song and Kou (2015) \cite{CSK} showed that this approach can be simplified by discretizing the state space and approximating the asset price process with a continuous-time Markov chain (CTMC). The CTMC approach was applied also in \cite{Chatterjee2018}  and was simplified further in Cui et al. (2018) \cite{Cui2018}. 

A precise approach for pricing Asian options under Monte Carlo (MC) simulations
uses an optimal importance sampling method which applies a change of drift such that the variance of the MC estimators is minimized. An asymptotically optimal change of drift can be determined by large deviations theory and was found for the Black-Scholes model by Glasserman et al. (1999) \cite{Glasserman1999} in discrete time and by Guasoni and Robertson (2008) \cite{Guasoni2008} in continuous time.

When the log asset price has independent increments, a backward recursion combined with 
Fourier inversion methods can be used in pricing Asian options. This was applied in the context of the Black-Scholes (BS) model by 
Carverhill and Clewlow (1990) \cite{CarvClewlow}, and was improved by Benhamou (2002) \cite{Benhamou}. 
The method was applied also to 
exponential L\'{e}vy models by Fusai (2004) \cite{Fusai2004} and Fusai and Meucci (2008) \cite{Fusai2008}.
An alternative recursion method has been presented in \cite{PZIME} for pricing discretely sampled Asian options in the Black-Scholes model.

On the other hand, Asian option pricing under stochastic volatility models has been much less studied.
Fouque and Han (2003) \cite{Fouque2003} studied the pricing of Asian options in a class of models where the volatility is driven by several mean-reverting processes. 
Robertson (2010) \cite{Robertson2010} extended the optimal sampling results of Guasoni and Robertson (2008) \cite{Guasoni2008} to stochastic volatility models, with explicit results for the Heston model. 
In recent work Al\`os et al. (2024) \cite{Alos2022} obtained an expansion for the implied volatility of Asian options around the ATM point for stochastic volatility models with Gaussian and fractional volatility, using Malliavin calculus methods. Lin and Chang (2020) \cite{Lin2020} obtained an expansion in volatility of volatility for Asian options in stochastic volatility models. 
Fusai and Kyriakou (2016) \cite{Fusai2016} derived a lower bound approximation for Asian option prices that works
for a wide class of models including stochastic volatility models.

We study in this paper the short-maturity asymptotics of Asian option prices in local-stochastic volatility models. The asymptotics of Asian options in the short-maturity limit has been studied in the literature for the local volatility model \cite{PZAsian, PZAsianCEV, Arguin2018,Park2019}, which includes as a limiting case the Black-Scholes model when the local volatility is a constant, and for jump-diffusion models \cite{PZJD}. In the Black-Scholes model the subleading $O(T)$ correction has been computed in \cite{subleading}.
A double expansion in maturity and log-strike was obtained by PDE methods for the prices of Asian options in the local volatility model in \cite{FPP2013}. The numerical performance of the short-maturity expansion for pricing Asian options in the Black-Scholes model with realistic maturities is comparable or better to that of alternative methods which are more time-intensive; see \cite{PZRisk} for a detailed comparison.

The asymptotics for the local volatility model is different in the two regimes of out-of-the-money (OTM) and at-the-money (ATM) Asian options.
For OTM Asian options, the leading asymptotics is obtained using large deviations theory methods 
\cite{Dembo1998,VaradhanLD}, while the ATM asymptotics is dominated by Gaussian fluctuations of the asset price around the spot price. 
We mention also the results of 
Gobet and Miri (2014) \cite{GobetMiri} who studied the expansion of the time-average of a diffusion in a small parameter such as time or volatility using Malliavin calculus, with applications to pricing Asian options. 

In this paper we consider both OTM and ATM short-maturity asymptotics for Asian options with continuous time averaging in local-stochastic volatility models. This includes as a particular case the stochastic volatility models. The asymptotics for ATM case has a simple leading-order term 
and the asymptotics for the OTM case is expressed in terms of a rate function which is given by the solution of a two-dimensional variational problem, 
which does not seem to yield simple explicit expressions. 
Instead, we develop a novel expansion method for the variational problem by expanding the rate function 
around the ATM point. In particular, we solve the variational problem in an expansion in log-moneyness $x:=\log(K/S_{0})$ and obtain explicit results for the first three terms in the expansion of the rate function in $x$. 
To the best of our knowledge, our expansion method is the first in the literature to approximate
the variational problem that arises from a rate function associated with the OTM short-maturity asymptotics
for option pricing by expanding around the ATM point without solving the variational problem first. 
We believe that our expansion method can be applicable in many other problems
in OTM option pricing where the rate function derived from large deviations theory cannot be explicitly solved.
For the limiting case of the log-normal SABR model we reproduce the first two terms which were obtained recently by Al\`os et al. (2024) \cite{Alos2022}.
We present also comparisons with numerical simulations for SABR, Heston and 
a local-stochastic volatility model, which show good agreement with the asymptotic expansion for Asian options with sufficiently small maturities and strikes sufficiently close to the ATM point. 

The paper is organized as follows. We introduce the local-stochastic volatility model 
and the technical assumptions on the model coefficients in Section~\ref{sec:model}.  
Section~\ref{sec:main} contains the main results of the paper:  Theorem \ref{thm:OTM} gives
the short-maturity asymptotics of OTM Asian options, and Theorem \ref{thm:ATM} gives the corresponding asymptotics for the ATM case.
For the OTM case, we develop a novel expansion around
the ATM point which is presented in Proposition \ref{prop:first:order}.
The limits of perfectly correlated and anticorrelated asset price and volatility ($\rho = \pm 1$) are more tractable and an explicit solution is obtained for these cases in Proposition \ref{prop:pm:1}. 
Applications to SABR, Heston models and numerical experiments 
are presented in Section~\ref{sec:numerics}.
Finally, we provide some background of large deviations theory in Appendix~\ref{sec:LDP},
and present the results for floating-strike Asian options in Appendix~\ref{sec:floating}.
The technical proofs of the results are given in Appendix~\ref{sec:proofs},
and Appendix~\ref{sec:appC} gives certain special functions used in the proof of Proposition  \ref{prop:first:order}.

%%%%%%%%%%%%%%%%%%%%%%%%%%
\section{Local-Stochastic Volatility Model}\label{sec:model}

Suppose that the underlying asset $S_{t}$ under the risk-neutral probability measure $\mathbb{Q}$ 
follows a local-stochastic volatility model of the form:
\begin{align}
&\frac{dS_{t}}{S_{t}}=(r-q)dt+\eta(S_{t})\sqrt{V_{t}}dB_{t},\label{eqn:S}
\\
&\frac{dV_{t}}{V_{t}}=\mu(V_{t})dt+\sigma(V_{t})dZ_{t},\label{eqn:V}
\end{align}
with $S_{0},V_{0}>0$, where $B_{t},Z_{t}$ are two standard Brownian motions that are correlated with correlation $\rho$, 
and $r\geq 0$ is the risk-free rate and $q\geq 0$ is the dividend yield. 
When $\eta(\cdot)$ is constant, \eqref{eqn:S}-\eqref{eqn:V} reduces
to the stochastic volatility model;
when $\mu(\cdot)$ and $\sigma(\cdot)$ are constantly zero so that $V_{t}\equiv V_{0}$, 
\eqref{eqn:S}-\eqref{eqn:V} reduces
to the local volatility model.
We assume that $\eta(\cdot),\sigma(\cdot):\mathbb{R}_{+}\rightarrow\mathbb{R}_{+}$ and $\mu(\cdot):\mathbb{R}_{+}\rightarrow\mathbb{R}$ are
uniformly bounded for simplicity, that is, we impose the following assumption.

\begin{assumption}\label{assump:bounded}
We assume that $\eta(\cdot),\mu(\cdot)$ and $\sigma(\cdot)$ are
uniformly bounded:
\begin{equation}
\sup_{x\in\mathbb{R}_{+}}\eta(x)\leq M_{\eta},
\qquad
\sup_{x\in\mathbb{R}_{+}}|\mu(x)|\leq M_{\mu},
\qquad
\sup_{x\in\mathbb{R}_{+}}\sigma(x)\leq M_{\sigma}.
\end{equation}
\end{assumption}

In addition, we impose the following
assumption on $\eta(\cdot)$ and $\sigma(\cdot)$
that is needed for the small-time large deviations estimates for \eqref{eqn:S}-\eqref{eqn:V}.

\begin{assumption}\label{assump:LDP}
We assume that $\inf_{x\in\mathbb{R}_{+}}\sigma(x)>0$
and $\inf_{x\in\mathbb{R}_{+}}\eta(x)>0$. Moreover, 
there exist some constants $M,\alpha>0$ such that
for any $x,y\in\mathbb{R}$,
$|\sigma(e^{x})-\sigma(e^{y})|\leq M|x-y|^{\alpha}$
and $|\eta(e^{x})-\eta(e^{y})|\leq M|x-y|^{\alpha}$.
\end{assumption}

We also assume that $p$-th moment for the asset price is finite for some $p>1$.

\begin{assumption}\label{assump:S:T:p}
There exists some $p>1$, such that there exists some $C'_{p}\in(0,\infty)$, such that $\mathbb{E}[S_{t}^p]\leq C'_{p}$ 
for any sufficiently small $t>0$.
\end{assumption}

Note that Assumption~\ref{assump:S:T:p} is satisfied
under mild conditions. For example, it is satisfied
when $\rho < - \sqrt{(p-1)/p}$ for some $p>1$ and Assumption~\ref{assump:bounded} holds, where $\rho$ is the correlation between the Brownian motions $B_{t},Z_{t}$ in \eqref{eqn:S}-\eqref{eqn:V}; see
Remark 2.1. in \cite{PWZ2024} for the details, where short-maturity VIX and European options under the local-stochastic volatility model \eqref{eqn:S}-\eqref{eqn:V} are studied \cite{PWZ2024}.

For Asian options, the prices of the call
and put options are given 
by\footnote{These are the so-called fixed-strike Asian options. 
Alternative instruments called floating-strike Asian options are discussed in Appendix \ref{sec:floating}.}
\begin{equation}\label{asian:price:defn}
C(T)=e^{-rT}\mathbb{E}\left[\left(\frac{1}{T}\int_{0}^{T}S_{s}ds-K\right)^{+}\right],
\qquad
P(T)=e^{-rT}\mathbb{E}\left[\left(K-\frac{1}{T}\int_{0}^{T}S_{s}ds\right)^{+}\right],
\end{equation}
where $T>0$ is the maturity, $K>0$ is the strike price and $(S_{t})_{t\geq 0}$ is the asset price process
that satisfies \eqref{eqn:S}.

The forward price of an Asian option with maturity $T$ is given by
\begin{equation}
F(T) := \mathbb{E}\left[\frac{1}{T} \int_0^T S_t dt \right] = S_0 \frac{e^{(r-q) T}-1}{r-q} \,.
\end{equation}

We distinguish between the three cases: the out-of-the-money (OTM) case 
when $K > F(T)$ for call options and $K< F(T)$ for put options, and
the in-the-money (ITM) case when $K < F(T)$ for call options
and $K > F(T)$ for put options, 
and the at-the-money (ATM) case when $K = F(T) $ for both call and put options.

In the short-maturity limit we have $\lim_{T\to 0} F(T) = S_0$ such that in this limit the three cases are distinguished by the relation of the strike $K$ to the spot price $S_0$.
We will study the OTM case and the ATM case, whereas
the ITM case can be studied by the put-call parity via the OTM case
and is hence omitted here.

%%%%%%%%%%%%%%%%%%%%%%%%%%%%%%%%%%%

\section{Main Results}\label{sec:main}

\subsection{Asymptotics for OTM Asian options}

In this section we consider the case of OTM Asian options. Recall that in the short-maturity
limit this corresponds to $S_{0}<K$ for Asian call options and $S_{0}>K$ for Asian put options.

\begin{theorem}\label{thm:OTM}
Suppose Assumptions~\ref{assump:bounded}, \ref{assump:LDP} and \ref{assump:S:T:p} hold.
When $S_{0}<K$, the Asian call options are OTM and we have
\begin{equation}
\lim_{T\rightarrow 0}T\log C(T)=-\mathcal{I}_{\rho}(S_{0},V_{0},K),
\end{equation}
and when $S_{0}>K$, 
the Asian put options are OTM and we have
\begin{equation}
\lim_{T\rightarrow 0}T\log P(T)=-\mathcal{I}_{\rho}(S_{0},V_{0},K),
\end{equation}
where
\begin{align}
\mathcal{I}_{\rho}(S_{0},V_{0},K)
:=\inf_{\substack{g(0)=\log S_{0},h(0)=\log V_{0}\\
\int_{0}^{1}e^{g(t)}dt=K, g,h\in\mathcal{AC}[0,1]}}\Lambda_{\rho}[g,h],\label{I:rho:1}
\end{align}
where $\mathcal{AC}[0,1]$ is the space of absolutely continuous functions on $[0,1]$ and
\begin{align}
\Lambda_{\rho}[g,h]:=\frac{1}{2(1-\rho^{2})}\int_{0}^{1}\left(\frac{g'(t)}{\eta(e^{g(t)})\sqrt{e^{h(t)}}}-\frac{\rho h'(t)}{\sigma(e^{h(t)})}\right)^{2}dt
+\frac{1}{2}\int_{0}^{1}\left(\frac{h'(t)}{\sigma(e^{h(t)})}\right)^{2}dt.\label{I:rho:2}
\end{align}
\end{theorem}

Any optimal $g$ and $h$ for the variational problem \eqref{I:rho:1} must satisfy the
Euler-Lagrange equations. The solutions of the Euler-Lagrange equations determine stationary points or extremals for the variational problem.

We take into account the constraint on $g$ by introducing the functional
\begin{align}
\Lambda_{\rho}[g,h;\lambda]&:=\frac{1}{2(1-\rho^{2})}\int_{0}^{1}\left(\frac{g'(t)}{\eta(e^{g(t)})\sqrt{e^{h(t)}}}-\frac{\rho h'(t)}{\sigma(e^{h(t)})}\right)^{2}dt
+\frac{1}{2}\int_{0}^{1}\left(\frac{h'(t)}{\sigma(e^{h(t)})}\right)^{2}dt
\nonumber
\\
&\qquad
+\lambda\int_{0}^{1}e^{g(t)}dt
\nonumber
\\
&=\frac{1}{2(1-\rho^{2})}\int_{0}^{1}\left(\frac{g'(t)}{\eta(e^{g(t)})\sqrt{e^{h(t)}}}\right)^{2}dt
+\frac{1}{2(1-\rho^{2})}\int_{0}^{1}\left(\frac{h'(t)}{\sigma(e^{h(t)})}\right)^{2}dt
\nonumber
\\
&\qquad
-\frac{\rho}{1-\rho^{2}}\int_{0}^{1}\frac{g'(t)}{\eta(e^{g(t)})\sqrt{e^{h(t)}}}\frac{h'(t)}{\sigma(e^{h(t)})}dt
+\lambda\int_{0}^{1}e^{g(t)}dt,
\end{align}
where $\lambda$ is the Lagrange multiplier.
The Euler-Lagrange equations satisfied by $g$ and $h$ for the variational problem \eqref{I:rho:1} are given by
$\frac{\partial \Lambda_{\rho}}{\partial g}=\frac{d}{dt}\left(\frac{\partial \Lambda_{\rho}}{\partial g'}\right)$
and $\frac{\partial \Lambda_{\rho}}{\partial h}=\frac{d}{dt}\left(\frac{\partial \Lambda_{\rho}}{\partial h'}\right)$.
This leads to the system of coupled ordinary differential equations for $g$ and $h$
\begin{align}
&\frac{d}{dt}\left(\frac{1}{1-\rho^{2}}\frac{g'(t)}{\eta^{2}(e^{g(t)})e^{h(t)}}
-\frac{\rho}{1-\rho^{2}}\frac{h'(t)}{\eta(e^{g(t)})\sqrt{e^{h(t)}}\sigma(e^{h(t)})}\right)
\nonumber
\\
&=-\frac{1}{1-\rho^{2}}\frac{(g'(t))^{2}\eta'(e^{g(t)})e^{g(t)}}{\eta^{3}(e^{g(t)})e^{h(t)}}
+\frac{\rho}{1-\rho^{2}}\frac{g'(t)h'(t)\eta'(e^{g(t)})e^{g(t)}}{\eta^{2}(e^{g(t)})\sqrt{e^{h(t)}}\sigma(e^{h(t)})}
\nonumber
\\
&\qquad\qquad\qquad
+\lambda e^{g(t)},\label{EL:1}
\end{align}
and
\begin{align}
&\frac{d}{dt}\left(\frac{1}{1-\rho^{2}}\frac{h'(t)}{\sigma^{2}(e^{h(t)})}
-\frac{\rho}{1-\rho^{2}}\frac{g'(t)}{\eta(e^{g(t)})\sqrt{e^{h(t)}}\sigma(e^{h(t)})}\right)
\nonumber
\\
&=-\frac{1}{2(1-\rho^{2})}\frac{(g'(t))^{2}}{\eta^{2}(e^{g(t)})e^{h(t)}}
-\frac{1}{1-\rho^{2}}\frac{(h'(t))^{2}\sigma'(e^{h(t)})e^{h(t)}}{\sigma^{3}(e^{h(t)})}
\nonumber
\\
&\qquad
+\frac{\rho}{1-\rho^{2}}\frac{g'(t)}{\eta(e^{g(t)})\sqrt{e^{h(t)}}}\frac{h'(t)\sigma'(e^{h(t)})e^{h(t)}}{\sigma^{2}(e^{h(t)})}
+\frac{\rho}{2(1-\rho^{2})}\frac{g'(t)}{\eta(e^{g(t)})\sqrt{e^{h(t)}}}\frac{h'(t)}{\sigma(e^{h(t)})},\label{EL:2}
\end{align}
with the constraints $g(0)=\log S_{0}$, $h(0)=\log V_{0}$ and $\int_{0}^{1}e^{g(t)}dt=K$.
The transversality condition gives $g'(1)=h'(1)=0$.

The problem of the existence and uniqueness of the solutions of the Euler-Lagrange equations for arbitrary strike $K$ will not be discussed here. This requires a separate study, under appropriate conditions on the functions $\eta(x),\sigma(v)$. In the next section we restrict to near-ATM strikes $K$, and will prove the existence and uniqueness of the solutions of these equations by direct construction.

\subsection{Expansion in log-moneyness}
\label{sec:3.2}

In general, it is hard to solve the system of coupled differential equations (\ref{EL:1}), (\ref{EL:2}) resulting from the Euler-Lagrange equations for the optimal paths $g(t), h(t)$ in closed form. 
However, we know that for the ATM case, i.e. $K=S_{0}$, we have
$\mathcal{I}_{\rho}(S_{0},V_{0},K)=0$
corresponding to the optimal paths $g(t)\equiv g_{0}:= \log S_{0}$ and $h(t)\equiv h_{0}:=\log V_{0}$. 

Therefore it is reasonable to search for a solution of these equations in an expansion around the ATM point. We present next an expansion for the
rate function $\mathcal{I}_{\rho}(S_{0},V_{0},K)$ following from an expansion of the optimal paths in powers of the log-moneyness $x:=\log(K/S_{0})$. 

We will derive the first three terms in this expansion. The result is formulated in terms of the coefficients in the expansion of the local volatility function around $S_0$
\begin{equation}\label{eta:expansion}
\eta(S) = \eta_0 + \eta_1 \log \frac{S}{S_0} + \eta_2 \log^2 \frac{S}{S_0} + O\left(\log^3(S/S_0)\right),
\end{equation}
and analogous for the expansion of the volatility-of-volatility function around $V_0$
\begin{equation}\label{sigma:expansion}
\sigma(V) = \sigma_0 + \sigma_1 \log \frac{V}{V_0} + \sigma_2 \log^2 \frac{V}{V_0} 
+ O(\log^3(V/V_0)) \,.
\end{equation}
More explicitly, we have 
\begin{equation}\label{etadef}
\eta_0 = \eta(S_0)\,,\quad \eta_1=S_0 \eta'(S_0)\,, \quad
\eta_2 = \frac12 S_0 \eta'(S_0) + \frac12 S_0^2 \eta''(S_0)\,,
\end{equation}
 and analogous for the volatility-of-volatility parameters
 \begin{equation}
 \sigma_0=\sigma(V_0)\,,\quad
 \sigma_1 = V_0 \sigma'(V_0)\,,\quad
 \sigma_2 = \frac12 V_0 \sigma'(V_0) + \frac12 V_0^2 \sigma''(V_0)\,.
 \end{equation}

\begin{proposition}\label{prop:first:order}
Assume that $\eta(x),\sigma(v)$ are twice continuously differentiable such that the expansions (\ref{eta:expansion}) and (\ref{sigma:expansion}) are valid.
Then we have the following expansion of the rate function in powers of 
log-moneyness $x=\log(K/S_0)$:
\begin{align}\label{JAexp}
 \mathcal{I}_{\rho}(S_{0},V_{0},S_{0}e^{x})
& =\frac{3}{2\eta_0^{2} V_{0}}x^{2} 
- \frac{3}{10 \eta_0^3 V_0^{3/2}} \left( 3 \rho \sigma_0 + (\eta_0 + 6 \eta_1) \sqrt{V_0} \right) x^3 \\
& \qquad\qquad\qquad+ %\frac{1}{1400\eta_0^4 V_0^2} 
\frac{\beta_0 V_0 
+ \beta_1 \sigma_0  +  \beta_2  \sigma_0^2}{1400\eta_0^4 V_0^2}  x^4  + O(x^{5}) \,, \nonumber
\end{align}
as $x\rightarrow 0$, with
\begin{align}
& \beta_0 := 109 \eta_0^2 + 2664 \eta_1^2 + 36\eta_0 (13 \eta_1 - 60 \eta_2), \\
& \beta_1 := 18\rho \left(-30 \rho \sigma_1 + (13\eta_0 + 18\eta_1) \sqrt{V_0} \right), \\
& \beta_2 := 9(-25 + 99 \rho^2), 
\end{align}
where $\eta_{0,1,2}$ are defined in \eqref{eta:expansion} and $\sigma_{0,1,2}$ are defined in \eqref{sigma:expansion}.

The optimal $g,h$ that solve the variational problem \eqref{I:rho:1} admit 
the expansion:
\begin{align}
&g(t)=g_{0}(t)+xg_{1}(t)+ x^2 g_2(t) + x^{3}g_{3}(t)+O(x^{4}),\nonumber
\\
&h(t)=h_{0}(t)+xh_{1}(t)+x^2 h_2(t) + x^{3}h_{3}(t)+O(x^{4}),\label{g:h:expansion:first}
\end{align}
where $g_{0}(t)\equiv\log S_{0}$, $h_{0}(t)\equiv\log V_{0}$ and
\begin{align}
g_{1}(t)=\frac{3}{2}(2t-t^{2}),
\qquad
h_{1}(t)=\frac{3\rho\sigma_0}{2\eta_0\sqrt{V_{0}}}(2t-t^{2}),
\qquad
0\leq t\leq 1\,,
\end{align}
and the functions $g_{2,3}(t), h_{2,3}(t)$ are given in the proof of this result and in Appendix~\ref{sec:proofs}. 
\end{proposition}

\begin{remark}
Although the coefficient $\sigma_1$ in \eqref{sigma:expansion} 
appears in intermediate quantities, such as the function $h_2(t)$, it disappears in the final result for the $O(x^3)$ term in the rate function (\ref{JAexp}). A similar cancellation takes place in the $O(x^4)$ term, which does not depend on $\sigma_2$, although this parameter appears in $h_3(t)$.
\end{remark}

%%%%%%%%%%%%%%%%%%%%%%%%%%%%%%%%%%%%%%%%%%

\subsection{Perfectly correlated and anticorrelated cases}
\label{sec:3.3}

One limiting case where an exact solution of the variational problem for the rate function can be found in closed form is the perfectly correlated case ($\rho = 1$) and perfectly anticorrelated case ($\rho=-1$).
For this case we have the following result.

\begin{proposition}\label{prop:pm:1}
\begin{equation}\label{pm:variational}
\mathcal{I}_{\pm 1}(S_{0},V_{0},K)
=\inf_{\substack{h(0)=\log V_{0}\\
\int_{0}^{1}\mathcal{F}_{\pm}(e^{h(t)})dt=K, h\in\mathcal{AC}[0,1]}}\frac{1}{2}\int_{0}^{1}\left(\frac{h'(t)}{\sigma(e^{h(t)})}\right)^{2}dt,
\end{equation}
where $\mathcal{F}_{\pm}(\cdot)$ is defined as:
\begin{equation}\label{pm:F}
\int_{S_{0}}^{\mathcal{F}_{\pm}(x)}\frac{dy}{y\eta(y)}=\int_{V_{0}}^{x}\frac{\pm dy}{\sqrt{y}\sigma(y)},\qquad\mbox{ for any } x>0 \,.
\end{equation}
\end{proposition}

We note that the variational problem \eqref{pm:variational}
can be solved analytically. This can be seen by casting it into the same form as the variational problem for 
Asian options in the local volatility model, which was solved explicitly in \cite{PZAsian}.
We reformulate the variational problem \eqref{pm:variational} by defining $g(t)=\log\mathcal{F}_{\pm}(e^{h(t)})$,
% eqref{pm:variational}, 
such that $h(t)=\log\mathcal{F}_{\pm}^{-1}(e^{g(t)})$, where $\mathcal{F}^{-1}_{\pm}$ is the inverse 
function of $\mathcal{F}_{\pm}$
which is monotonic by the definition in \eqref{pm:F} and hence the inverse function $\mathcal{F}^{-1}_{\pm}$ is well-defined.
Then, we can rewrite \eqref{pm:variational} as
\begin{equation}\label{pm:variational:2}
\mathcal{I}_{\pm 1}(S_{0},V_{0},K)
=\inf_{\substack{g(0)=\log\mathcal{F}_{\pm}(V_{0})\\
\int_{0}^{1}e^{g(t)}dt=K,g\in\mathcal{AC}[0,1]}}\frac{1}{2}\int_{0}^{1}\left(\frac{g'(t)}
{\hat{\sigma}_\pm(e^{g(t)})}\right)^{2}dt\,,
\end{equation}
where $\hat{\sigma}_\pm(S):=\frac{1}{S}\mathcal{F}_{\pm}^{-1}(S) |\mathcal{F}'_{\pm}(\mathcal{F}^{-1}_{\pm}(S)) |\sigma(\mathcal{F}^{-1}_{\pm}(S))$ for any $S>0$.
Then, \eqref{pm:variational:2} is exactly the rate function for Asian option
for local volatility models with the local volatility $\hat{\sigma}_\pm(\cdot)$, the spot asset price $\mathcal{F}_{\pm}(V_{0})$ and the strike price $K$; see \cite{PZAsian}.
The variational problem \eqref{pm:variational:2} has been solved analytically in Proposition~8 of \cite{PZAsian}.

%%%%%%%%%%%%%%%%%%%%%%%%
\subsection{Local volatility case}

When $\mu(\cdot)\equiv 0$, $\sigma(\cdot)\equiv 0$, $\rho=0$ and $V_{0}=1$, the model \eqref{eqn:S}-\eqref{eqn:V}
reduces to 
\begin{equation}\label{LV:SDE}
\frac{dS_{t}}{S_{t}}=(r-q)dt+\eta(S_{t})dB_{t},
\end{equation}
which is the local volatility model that was studied in \cite{PZAsian}.
We use the notation $\mathcal{I}_{\mathrm{LV}}(S_{0},K)$ to denote
the rate function in \eqref{I:rho:1} in this case, emphasizing
it corresponds to the local volatility model \eqref{LV:SDE}.
We recall from \eqref{I:rho:1} that
the rate function in \eqref{I:rho:1} is the infimum of $\Lambda_{0}[g,h]$, 
that is given in \eqref{I:rho:2}, 
subject to
$g(0)=\log S_{0}$, $h(0)=0$ and $\int_{0}^{1}e^{g(t)}dt=K$.
When $\sigma(\cdot)\equiv 0$
it follows from \eqref{I:rho:2}
that $\Lambda_{0}[g,h]=+\infty$ unless $h'(t)=0$
for every $0\leq t\leq 1$, which implies that the optimal $h$
is given by $h(t)=0$ for every $0\leq t\leq 1$. 
We conclude that
\begin{equation}\label{rate:local:vol}
\mathcal{I}_{\mathrm{LV}}(S_{0},K)
=\inf_{g(0)=\log S_{0},
\int_{0}^{1}e^{g(t)}dt=K,g\in\mathcal{AC}[0,1]}\frac{1}{2}\int_{0}^{1}\left(\frac{g'(t)}{\eta(e^{g(t)})}\right)^{2}dt,
\end{equation}
which recovers the result in \cite{PZAsian}.
Moreover, the variational problem in \eqref{rate:local:vol}
was solved analytically in \cite{PZAsian}.

We can verify explicitly that the result (\ref{JAexp}) reduces in the limit $\sigma(v) \equiv 0$
to the rate function for Asian options in the local volatility model that was obtained in \cite{PZAsian}.
By taking $\sigma_0 = \sigma_1 = 0$ in (\ref{JAexp}), the rate function becomes
\begin{align}
 \mathcal{I}_{\mathrm{LV}}(S_{0},S_{0}e^{x})
 & =\frac{3}{2\eta_0^{2} }x^{2} 
- \frac{3}{10 \eta_0^3 }  (\eta_0 + 6 \eta_1)   x^3 \\
& \qquad\qquad\qquad+
\frac{ 109 \eta_0^2 + 2664 \eta_1^2 + 36\eta_0 (13 \eta_1 - 60 \eta_2)}{1400\eta_0^4}  x^4  + O(x^{5}) \,. \nonumber
\end{align}
Substituting here the expressions (\ref{etadef}) this becomes
\begin{align}
 &\mathcal{I}_{\mathrm{LV}}(S_{0},S_{0}e^{x})
 \\
 & = \frac{1}{\eta_0^2} \Bigg\{ \frac32 x^2 - \frac{3}{10}\left(1 + 6 S_0 \frac{\eta'(S_0)}{\eta(S_0)}
\right) x^3\nonumber \\
&\qquad\qquad +
\left( \frac{109}{1400} + \frac{333}{175} \left( S_0 \frac{\eta'(S_0)}{\eta(S_0)} \right)^2 -
\frac{153}{350} S_0 \frac{\eta'(S_0)}{\eta(S_0)} - \frac{27}{35} S_0^2  \frac{\eta''(S_0)}{\eta(S_0)} 
 \right) x^4 + O(x^5) \Bigg\}\,. \nonumber
\end{align}

This reproduces precisely the first three terms of the rate function for Asian options in the local volatility model  given in Corollary 16 of \cite{PZAsian}, see equation (40) in \cite{PZAsian}.

%%%%%%%%%%%%%%%%%%%%%%%%%%%%%%%%%%%%%%%%%%
\subsection{ATM Case}

In this section, we consider the ATM case. 
Recall that in the short-maturity limit this corresponds to $K=S_{0}$
for both Asian call and put options.
We have the following result.

\begin{theorem}\label{thm:ATM}
Suppose Assumption~\ref{assump:bounded} holds.
We further assume that $\eta(\cdot)$ and $\sigma(\cdot)$ are Lipschitz
and there exists some $C'\in(0,\infty)$ 
such that $\max_{0\leq t\leq T}\mathbb{E}[(S_{t})^{4}]\leq C'$
for any sufficiently small $T>0$.
When $K=S_{0}$,
\begin{equation}
\lim_{T\rightarrow 0}\frac{C(T)}{\sqrt{T}}
=\lim_{T\rightarrow 0}\frac{P(T)}{\sqrt{T}}
=\frac{S_{0}\eta(S_{0})\sqrt{V_{0}}}{\sqrt{6\pi}}.
\end{equation}
\end{theorem}

We notice that unlike the OTM case, where the OTM option price decays to zero 
exponentially fast in terms of maturity $T$, for ATM case, the option prices scale
as $\sqrt{T}$ as $T\rightarrow 0$. This scaling is consistent with the ATM case
for Asian options for local volatility models studied in the literature \cite{PZAsian}.

%%%%%%%%%%%%%%%%
\section{Applications to Pricing Asian Options}\label{sec:numerics}

We present in this section the application of the asymptotic results to the practical problem of
pricing Asian options.
This is done most conveniently by introducing the \textit{equivalent log-normal 
volatility of an Asian option} $\Sigma_A(K,T;S_0)$ \cite{PZAsian}.  This is 
defined such that the price of the Asian option $C_A(K,T)$
with strike $K$ and maturity $T$ is
reproduced by substituting $\Sigma_A(K,S_0)$ into the Black-Scholes formula with an appropriately chosen forward price
\begin{equation}
C_A(K,T) = C_{\mathrm{BS}}(K, T;\Sigma_A(K,T),F(T))\,,
\end{equation}
where the forward price of the Asian option is
\begin{equation}
F(T) := \mathbb{E}\left[\frac{1}{T} \int_0^T S_t dt \right] = S_0 \frac{e^{(r-q) T}-1}{r-q} \,,
\end{equation}
and the Black-Scholes formula for the European call option is
\begin{equation}
C_{\mathrm{BS}}(K,T;\sigma,F) = e^{-rT} \left( F \Phi(d_1) - K \Phi(d_2) \right)\,,
\end{equation}
with $d_{1,2} := \frac{1}{\sigma \sqrt{T}} (\log(F/K) \pm \frac12 \sigma^2 T )$ and $\Phi(\cdot)$
the cumulative distribution function of a standard Gaussian distribution with mean $0$ and variance $1$.
A similar result holds for Asian put options
\begin{equation}
P_A(K,T) = P_{\mathrm{BS}}\left(K, T;\Sigma_A(K,T),F(T)\right)\,,
\end{equation}
where the Black-Scholes formula for the European put option is
\begin{equation}
P_{\mathrm{BS}}(K,T;\sigma,F) = e^{-rT} \left( K \Phi(- d_2) - F \Phi( -d_1) \right)\,.
\end{equation}

Note that with this definition the Asian put and call options satisfy the correct put-call parity
\begin{equation}
C_A(K,T) - P_A(K,T) = e^{-rT} (F(T) - K) \,.
\end{equation}

\begin{remark}
There is an alternative way of defining an implied volatility for an Asian option which was introduced in Section~4 of \cite{PZAsian} and was analyzed further in \cite{PZRisk}. This is that constant volatility $\sigma_{\rm imp}(K,T)$ which reproduces the Asian option price when assuming that the underlying asset follows a Black-Scholes model with volatility $\sigma_{\rm imp}(K,T)$. 

For practical applications the use of the implied volatility $\sigma_{\rm imp}(K,T)$ requires pricing an Asian option under the Black-Scholes model, while the equivalent log-normal volatility $\Sigma_A(K,T)$ requires just the Black-Scholes formula for European options, and is thus easier to use. Therefore we will focus in this paper on the equivalent log-normal volatility. Furthermore, their short-maturity limits are related, as shown in Proposition~17 of \cite{PZAsian}.
\end{remark}

For simplicity we will refer to the equivalent log-normal volatility as the implied volatility of an Asian option, since there is no possibility of confusion with the implied volatility of an Asian option $\sigma_{\rm imp}(K,T)$ discussed in \cite{PZAsian}. 

The short-maturity asymptotics for the Asian options studied in this paper translates to a small-time limit for the Asian implied volatility 
\begin{equation}\label{SigAdef}
\lim_{T\to 0} \Sigma_A(K,T) := \Sigma_A(K, S_0) = 
\frac{|\log (K/S_0) |}{\sqrt{2 \mathcal{I}_\rho(S_0,V_0,K)}} \,,
\end{equation}
which is expressed simply in terms of the rate function, as shown. Substituting here the result 
(\ref{JAexp}) for the rate function gives a corresponding expansion for the 
asymptotic Asian implied volatility in powers of log-moneyness
\begin{align}\label{SigAgen}
\Sigma_A(K,S_0) &= \sqrt{\frac{V_0}{3}}
\left\{ 1 + \frac{3\rho \sigma_0 + (\eta_0 + 6\eta_1) \sqrt{V_0}}{10\sqrt{V_0} \eta_0} x 
+  \frac{ \gamma_2 \sigma_0^2 + \gamma_1 \sigma_0 + \gamma_0}{4200 \eta_0^2 V_0}  x^2 + O(x^3) \right\}\,,
\end{align}
with
\begin{align}
\gamma_0 &:= -46\eta_0^2 - 396 \eta_1^2 + 288\eta_0 \eta_1 + 2160 \eta_0 \eta_2\,, \\
\gamma_1 &:= 36\rho \left(15\rho \sigma_1 + 2(2\eta_0 + 27\eta_1) \sqrt{V_0} \right)\,,\\
\gamma_2 &:= 225 - 324 \rho^2 \,.
\end{align}

In conclusion, for practical purposes of pricing an Asian option, knowledge of the rate function 
$\mathcal{I}_\rho(S_0,V_0,K)$ yields, by (\ref{SigAdef}), the small-maturity limit of the equivalent log-normal volatility of the Asian option $\Sigma_A(K,S_0)$ with the appropriate strike $K$. The price of the Asian option is obtained by substituting $\Sigma_A(K,S_0)$ into the Black-Scholes formula. In the remainder of this section we illustrate this application on the example of the log-normal SABR model, and will compare the asymptotic result for $\Sigma_A(K,S_0)$ with the result of a numerical Monte Carlo simulation.

%%%%%%%%%%%%%%%%%%%%%%%%%
\subsection{SABR model}

One of the simplest stochastic volatility models is the log-normal SABR model \
\begin{equation}\label{SABR:model}
\frac{dS_t}{S_t} = (r-q)dt+\sqrt{V_t} dB_t \,,\quad dV_t = \sigma V_t dZ_t\,,
\end{equation}
where $(B_t, Z_t)$ are two standard Brownian motions correlated with correlation $\rho$, $\sigma$ is a positive constant, 
$r$ is the risk-free rate and $q$ is the dividend yield.
We discuss in this section the implications of our results for the short-maturity asymptotics of Asian options in the log-normal SABR model.

The expansion of the Asian rate function in the SABR model can be obtained from Proposition \ref{prop:first:order} by taking $\eta_0=1, \eta_k=0$ and $\sigma_0=\sigma, \sigma_k=0$ for all 
$k\geq 1$.  The first three terms are
\begin{align}\label{JA:SABR}
\mathcal{I}_{\rho}(S_0,V_0,S_0e^x) 
&= \frac{3}{2V_0} x^2 - \frac{3(3\rho \sigma + \sqrt{V_0})}{10 V_0^{3/2}} x^3 
\nonumber
\\
&\qquad\qquad\qquad
+ \frac{9( -25  + 99\rho^2) \sigma^2 + 234 \rho \sigma \sqrt{V_0} + 109 V_0}
{1,400 V_0^2}x^{4} +
O(x^5).
\end{align}
The short-maturity limit of the Asian implied volatility is obtained from (\ref{SigAgen})
\begin{align}
\Sigma_A(K,S_0) &= \frac{\sqrt{V_0}}{\sqrt3} \Bigg( 1 + \left(\frac{1}{10}
+ \frac{3\sigma \rho}{10 \sqrt{V_0}} \right) x \nonumber
\\
&\qquad\qquad
+ \frac{9(25-36\rho^2) \sigma^2 + 144 \rho \sigma \sqrt{V_0} - 46 V_0}{4,200 V_0} x^2 +
O(x^3) \Bigg).\label{SABRSigA}
\end{align}
The $O(x)$ term gives the ATM skew of the Asian implied volatility, and agrees with the result 
in Al\'os et al. \cite{Alos2022}.
The $O(x^2)$ term gives the ATM convexity of the Asian implied volatility. It reproduces the result in \cite{PZAsian} for the Black-Scholes model by taking 
$\sigma = 0$.
Recall that the expansion of the Asian implied volatility in the Black-Scholes model is (see equation (55) in \cite{PZAsian})
\begin{equation}
\Sigma_A^{\mathrm{BS}}(K) = \frac{1}{\sqrt3} \sqrt{V_0} \left( 1 + \frac{1}{10} x - \frac{23}{2100} x^2 + O(x^3)
\right)\,. 
\end{equation}
The coefficient of the $O(x^2)$ term in this expression is indeed 
reproduced by taking $\sigma=0$ in (\ref{SABRSigA}).

We verify next that the perfectly correlated and correlated cases $\rho = \pm 1$ are correctly reproduced by Proposition~\ref{prop:pm:1}. According to this result, the rate function for $\rho = \pm 1$ is the same as the rate function for Asian options in a local volatility model with local volatility: 
\begin{equation}
\hat{\sigma}_\pm(S):=\frac{1}{S}\mathcal{F}_{\pm}^{-1}(S) | \mathcal{F}'_{\pm}(\mathcal{F}^{-1}_{\pm}(S)) |\sigma(\mathcal{F}^{-1}_{\pm}(S))\,,
\end{equation}
where $\mathcal{F}_\pm(S)$ is defined in (\ref{pm:F}). For the SABR model we take $\eta(\cdot)\equiv 1$ and $\sigma(\cdot) \equiv \sigma$ in this relation, which gives
\begin{equation}
\mathcal{F}_\pm(x) = S_0 e^{\pm \frac{2}{\sigma} (\sqrt{x} - \sqrt{V_0})}\,,
\end{equation}
and a straightforward computation gives
\begin{equation}
\hat \sigma_\pm(S) = \left| \sqrt{V_0} \pm \frac{\sigma}{2} \log(S/S_0) \right| \,.
\end{equation}

The expansion of the rate function for Asian options 
in the local volatility model is given in Corollary~16 of \cite{PZAsian}; see equation (40). The first three terms depend only on the quantities
\begin{equation}
\hat \sigma_\pm(S_0) = \sqrt{V_0}\,,\quad
S_0 \hat \sigma'_\pm(S_0) = \pm \frac{\sigma}{2} \,,\quad
S_0^2 \hat \sigma''_\pm(S_0) =  \mp \frac{\sigma}{2}\,.
\end{equation}
Substituting these values into equation (40) of \cite{PZAsian} gives the rate function:
\begin{align}
\mathcal{I}_{\pm 1}(S_0,V_0, S_0 e^x) &=
\frac{1}{V_0} \Bigg\{ \frac32 x^2 + \left( - \frac{3}{10} \mp \frac{9\sigma}{10\sqrt{V_0}}\right) x^3 
\nonumber
\\
&\qquad\qquad\qquad\qquad
+ \left( \frac{109}{1400} \pm \frac{117 \sigma}{700\sqrt{V_0}} + \frac{333\sigma^2}{700 V_0}\right) x^4 + O(x^5) \Bigg\}\,,
\end{align}
which agrees with the limiting values of (\ref{JA:SABR}) for $\rho = \pm 1$.

%%%%%%%%%%%%%%%%%%%%%%%%
\subsection{Heston model}

In this section we consider the asymptotic predictions for Asian options in the Heston model
\begin{equation}
\frac{dS_t}{S_t} = (r-q)dt+\sqrt{V_t} dB_t \,,\quad dV_t = \kappa (\theta - V_t) dt + \xi \sqrt{V_t} dZ_t\,,
\end{equation}
where $(B_t, Z_t)$ are correlated standard Brownian motions with correlation $\rho$,  
$r$ is the risk-free rate, $q$ is the dividend yield,
and $\kappa,\theta,\xi$ are positive constants.

For this model we have $\sigma(v) = \xi v^{-1/2}$, and thus we get from \eqref{sigma:expansion} that
\begin{equation}
\sigma_0 = \xi/\sqrt{V_0}\,,\quad 
\sigma_1 = - \xi/(2\sqrt{V_0}) \,.
\end{equation}

Although $\sigma(v)$ is not bounded and Lipschitz, and does not satisfy our Assumptions~\ref{assump:bounded} and \ref{assump:LDP}, it was shown by Robertson (2010) \cite{Robertson2010} that this is not required for the existence of a large deviation principle for $\mathbb{Q}(\{(S_{tT}, V_{tT}),0\leq t\leq 1\}\in\cdot)$.
The sample-path large deviation principle for the joint distribution $\mathbb{Q}(\{(S_{tT}, V_{tT}),0\leq t\leq 1\}\in\cdot)$
was established in \cite{Robertson2010}
and is based on an earlier result that $\mathbb{Q}(\{V_{tT},0\leq t\leq 1\}\in\cdot)$ satisfies a large deviation principle, which was obtained in Donati-Martin et al. (2004) \cite{DonatiMartin2004}.
Moreover, for the Heston model, one can show that Assumption~\ref{assump:S:T:p} still holds, 
which implies that Theorem~\ref{thm:OTM} is still valid and hence the subsequent results including Proposition~\ref{prop:first:order}.

From Proposition~\ref{prop:first:order}, we deduce that the 
rate function for Asian options with strike $K$ is given by:
\begin{align}\label{JA:Heston}
\mathcal{I}_\rho(S_0,V_0,K) &= \frac{3}{2V_0} x^2 - \frac{3}{10 V_0^2} ( 3\rho \xi + V_0) x^3\nonumber \\
&\quad + \frac{1}{1400 V_0^3} \left( 109 V_0^2 + 18\rho \xi (15\rho \xi + 13 V_0) + 9\xi^2 (-25 + 99\rho^2) \right) x^2 + O(x^3)  \,. 
\end{align}
It follows that the implied volatility of an Asian option in the Heston model is given by:
\begin{align}
\Sigma_A(K) & = \sqrt{\frac{V_0}{3}}
\Bigg( 1 + \frac{V_0 + 3 \rho \xi}{10 V_0} x \nonumber\\
&\qquad\qquad\qquad+ \frac{1}{4200 V_0^2}
\Big( 46 V_0^2 - 144 \rho \xi V_0 - 225 \xi^2
+ 594 \rho^2 \xi^2 \Big) x^2 + O(x^3) \Bigg)\,. \label{HestonSigA}
\end{align}

This result is reproduced by Proposition~\ref{prop:pm:1} for the 
 perfectly correlated and correlated cases $\rho = \pm 1$.
 For the Heston model we take $\eta(x)\equiv 1$ and $\sigma(v) \equiv \sigma v^{-1/2}$ in the relation (\ref{pm:F}), which gives
\begin{equation}
\mathcal{F}_\pm(x) = S_0 e^{\pm \frac{2}{\sigma} (x - V_0)}\,,
\end{equation}
Substituting into the definition of $\hat\sigma_\pm(S)$ we get
\begin{equation}
\hat \sigma_\pm(S) = \sqrt{V_0 \pm \frac{\sigma}{2} \log(S/S_0)} \,.
\end{equation}

The expansion of the rate function for Asian options 
in the local volatility model is given in Corollary~16 of \cite{PZAsian}; see equation (40). The first three terms depend only on the quantities
\begin{equation}
\hat \sigma_\pm(S_0) = \sqrt{V_0}\,,\quad
S_0 \hat \sigma'_\pm(S_0) = \pm \frac{\sigma}{2\sqrt{V_0}} \,,\quad
S_0^2 \hat \sigma''_\pm(S_0) =  \mp \frac{\sigma}{2 V_0^{3/2}} \Big( V_0 \pm \frac12 \sigma \Big)\,.
\end{equation}
Substituting into equation (40) of \cite{PZAsian} gives the first three terms of the rate function
\begin{align}
\mathcal{I}_{\pm 1}(S_0, V_0, S_0 e^x) &= \frac{3}{2V_0} x^2 + 
\frac{1}{V_0} \Big( -\frac{3}{10} \mp \frac{9\sigma}{10 V_0} \Big) x^3 \\
& + \frac{1}{V_0} \Big( \frac{109}{1400} \pm \frac{117}{700} \cdot \frac{\sigma}{V_0} + \frac{117}{175} \cdot \frac{\sigma^2}{V_0^2}\Big) x^4 + O(x^5)
\end{align}
It is easy to see that this is reproduced by taking $\rho = \pm 1$ in the general expression (\ref{JA:Heston}) for the Asian
rate function in the Heston model.

%%%%%%%%%%%%%%%%%%%%%%%%%%%%%%%%%%%%%%%%%%%%%%
\subsection{Numerical tests}

For the numerical tests we will compare the results of a Monte Carlo (MC) simulation of 
the equivalent log-normal Asian volatility $\Sigma_A(K,S_0)$ with the asymptotic prediction.
We define a quadratic approximation for the Asian options implied volatility using the asymptotic predictions for the ATM volatility level $\Sigma_{\mathrm{ATM}}$, skew $s_A$ and convexity $\kappa_A$ of the Asian implied volatility
\begin{equation}\label{Th}
\Sigma_{A}(x) = \Sigma_{\mathrm{ATM}} + s_A x + \kappa_A x^2\,,\qquad x = \log(K/S_0) \,.
\end{equation}
These parameters were given above for the SABR and Heston model, and can be extracted from equation (\ref{SigAgen}) for a general LSV model.
We present in this section numerical tests for the asymptotic implied volatility of Asian options
for the log-normal SABR model, the Heston model, and a local-stochastic volatility model with log-normal volatility.
\vspace{0.5cm}

\textbf{Log-normal SABR model.}
For the numerical tests we assume the following SABR parameters
\begin{equation}\label{params2}
\sigma = 2.0\,,\quad V_0 = 0.1 \,.
\end{equation}
The correlation $\rho$ will be varied in the range $\{ -0.7, 0, +0.7\}$. 
The spot asset price is $S_0=1$, the risk-free rate is $r=0$ and the dividend yield is $q=0$. 

The coefficients $\Sigma_{\mathrm{ATM}}$, $s_A$ and convexity $\kappa_A$ for the 
log-normal SABR model \eqref{SABR:model} are shown in equation (\ref{SABRSigA}).
The numerical values of these coefficients are shown in Table \ref{tab:3models}.

We compare the asymptotic implied volatility $\Sigma_A(K)$ against a Monte Carlo simulation of Asian options in the SABR model. 
The MC simulation used the number of paths $N_{\mathrm{MC}}=100k$ and a time-discretization using $n=200$ time steps. The $V_t$ process is simulated exactly as a geometric Brownian motion, and $S_t$ is simulated using a forward-stepping Euler discretization.

The results are shown in Figure~\ref{Fig:1} where 
the asymptotic approximation (\ref{Th}) is shown as the solid curve and 
the red dots show the results of the MC simulation for Asian options with maturity 
$T=1/52$ (1 week). The agreement is reasonably good for strikes sufficiently close to the ATM point.

\begin{figure}[h]
\centering
\includegraphics[width=1.9in]{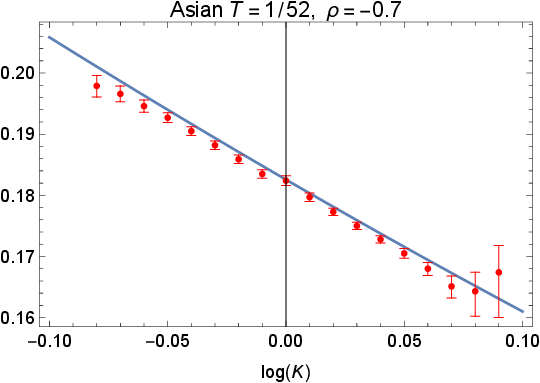}
\includegraphics[width=1.9in]{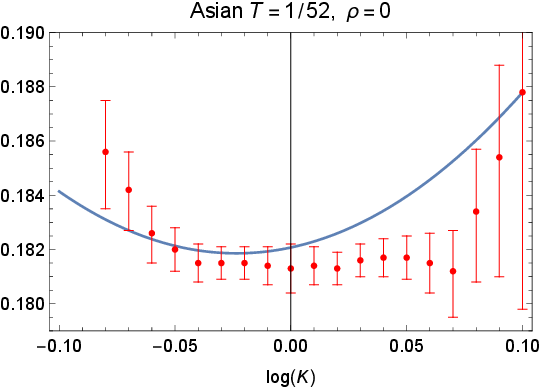}
\includegraphics[width=1.9in]{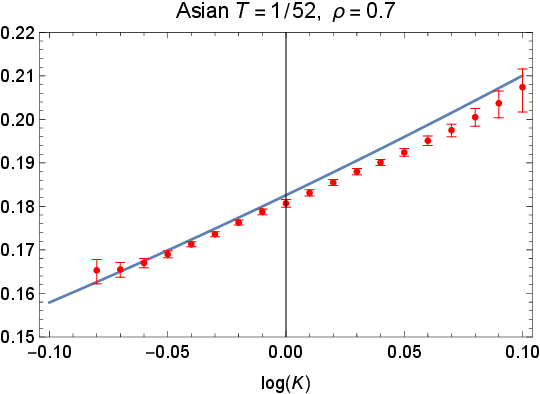}
\caption{Numerical tests for Asian options pricing in the SABR model with parameters 
(\ref{params2}) and correlation
$\rho\in \{ - 0.7, 0, +0.7\}$. The solid curve is the asymptotic prediction (\ref{Th}) for the implied volatility vs $x=\log(K/S_0)$ and
the red dots show the MC result for Asian options with maturity $T=1/52$ (1 week).}
\label{Fig:1}
\end{figure}

\begin{table}[h!]
  \centering
\begin{tabular}{|c|ccc||ccc||ccc|}
\hline
\multicolumn{1}{|c|}{}  & \multicolumn{3}{|c||}{SABR} & \multicolumn{3}{|c||}{Heston} 
                                   & \multicolumn{3}{|c|}{Tanh} \\    
\cline{2-10}    
%\hline
$\rho$ & $\Sigma_{\mathrm{ATM}}$ & $s_A$ & $\kappa_A$ 
           & $\Sigma_{\mathrm{ATM}}$ & $s_A$ & $\kappa_A$ 
           & $\Sigma_{\mathrm{ATM}}$ & $s_A$ & $\kappa_A$  \\
    \hline\hline
$-0.7$ & 0.183 & $-0.224$ & 0.085 
           & 0.115  & $-0.110$  & 0.060 
           & 0.183 & $-0.279$ & 0.149 \\
0 &    0.183 & 0.018 & 0.389
      & 0.115 & 0.011 & $-0.153$ 
      & 0.183 & $-0.036$ & 0.266 \\
$+0.7$ & 0.183 & 0.261 & 0.141
           & 0.115 & 0.133 & 0.033
           & 0.183 & 0.206 & $-0.170$  \\
\hline
\end{tabular}%
 \caption{The parameters for the short-maturity asymptotics of the Asian options for the three models used for the numerical tests. }
\label{tab:3models}%
\end{table}%

%\section*{Numerical tests: Heston model}

\textbf{Heston model.}
For the numerical tests in the Heston model we use the following parameter values
\begin{equation}\label{Hparams}
\kappa = 2.0\,,\quad
\theta = 0.09\,,\quad
\xi = 0.2\,,\quad V_0 = 0.04\,.
\end{equation}
The correlation will be varied in the range $\rho \in \{-0.7,0,+0.7\}$. The spot asset price is $S_0=1$, the risk-free rate is $r=0$ and the dividend yield is $q=0$.
The coefficients $\Sigma_{\mathrm{ATM}}$, $s_A$ and $\kappa_A$ in the quadratic volatility approximation (\ref{Th}) for the Heston model are shown in equation (\ref{HestonSigA}).

%We use the same parameters for the MC simulation: the number of MC paths is 100k, and the timeline is discretized with $n=200$ time steps. Both the $V_t$ and $S_t$ processes are simulated by a forward-stepping Euler discretization. 

The test results are shown in Figure~\ref{Fig:Heston}, where the solid curve is the asymptotic prediction for the Asian implied volatility and the red dots show the results of a MC simulation for Asian options with maturity $T=1/52$ (1 week). The asymptotic prediction is in good agreement with the MC simulation for strikes not too far away from the ATM point $x=0$. The parameters of the MC simulation are the same as for the SABR model. 
\vspace{0.5cm}

\begin{figure}[h]
\centering
\includegraphics[width=1.9in]{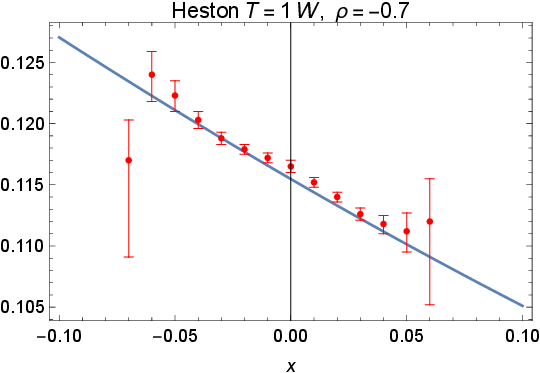} 
\includegraphics[width=1.9in]{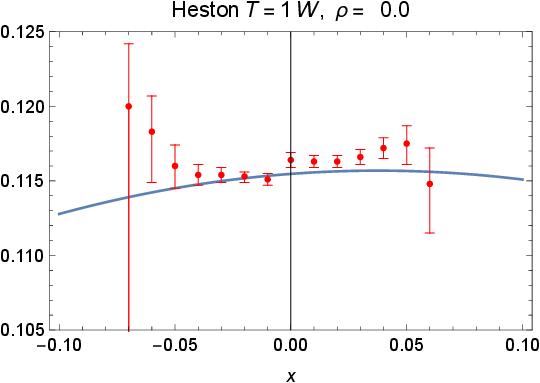}
\includegraphics[width=1.9in]{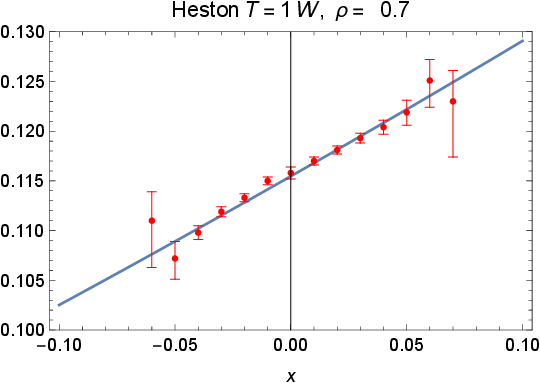}
\caption{Numerical tests for Asian options pricing in the Heston model with parameters (\ref{Hparams}) 
and correlation
$\rho \in \{ - 0.7, 0, +0.7\}$. The solid curve is the asymptotic prediction for the implied volatility vs $x = \log(K/S_0)$ and
the red dots show the MC result for Asian options with maturity $T=1/52$ (1 week).}
\label{Fig:Heston}
\end{figure}

\textbf{Tanh model.} Finally, we present also tests in a local-stochastic volatility model of the form:
\begin{equation}
dS_t/S_t = (r-q)dt+\eta(S_t) \sqrt{V_t} dB_t \,,\qquad dV_t/V_t= \sigma dZ_t\,,
\end{equation}
where $(B_t, Z_t)$ are correlated with correlation $\rho$, $r$ is the risk-free rate, $q$ is the dividend yield and $\sigma$ is a positive constant. The local volatility function is 
\begin{equation}\label{loc:vol:eta:S} 
\eta(S) = f_0 + f_1 \tanh(\log (S/S_0) - x_0) \,.
\end{equation}

This is the so-called Tanh model used by Forde and Jacquier \cite{Forde2011}. 
The local volatility function \eqref{loc:vol:eta:S} is expanded in powers of the log-asset 
$\log(S/S_0)$ as
\begin{equation}
\eta(S) = \eta_0 + \eta_1 \log\frac{S}{S_0} +  \eta_2 \log^2\frac{S}{S_0} + \cdots\,,
\end{equation}
with 
\begin{eqnarray}
\eta_0 = f_0 - f_1 \tanh x_0\,, \quad
\eta_1 = \frac{f_1}{\cosh^2 x_0}\,,  \quad
\eta_2 = \frac{f_1}{\cosh^2 x_0} \tanh x_0 \,.
\end{eqnarray}

We take the parameters $f_0=1,f_1=-0.5, x_0=0$.
For the $V_t$ process we assume the same parameters as above
$\sigma = 2.0$, $V_0 = 0.1$. The spot asset price is $S_0=1$, the risk-free rate is $r=0$ and the dividend yield is $q=0$.
The remaining parameters of the MC simulation are the same as in the other tests.

\begin{figure}[h]
\centering
\includegraphics[width=1.5in]{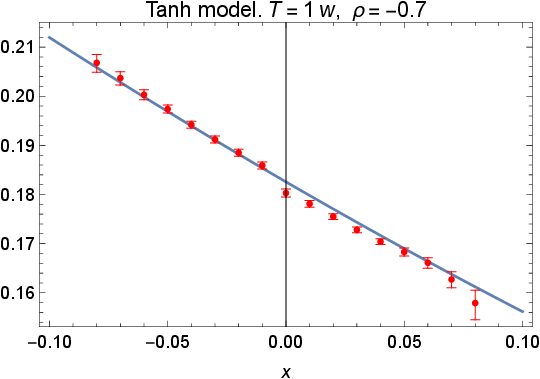}
\includegraphics[width=1.5in]{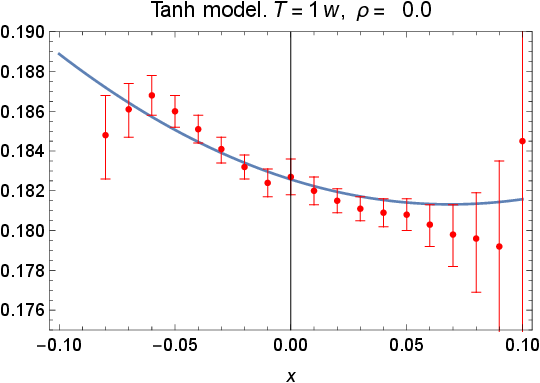}
\includegraphics[width=1.5in]{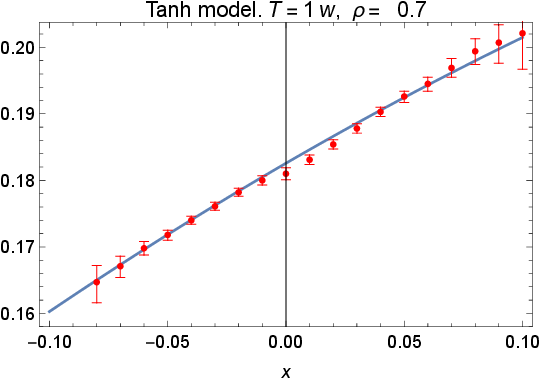}
\caption{Tests for Asian options with maturity $T=1/52$ (1 week) in the Tanh LSV model.}
\label{Fig:Tanh}
\end{figure}

The coefficients $\Sigma_{\mathrm{ATM}}, s_A,\kappa_A$ of the asymptotic Asian implied 
volatility for this model are taken from the general result (\ref{SigAgen}). 
The MC simulation results are shown in Figure \ref{Fig:Tanh} for Asian options with maturity $T=1/52$ 
(1 week) and $\rho \in \{ -0.7, 0, +0.7\}$. The results agree well with the asymptotic predictions (solid curve), which shows that the local volatility feature of the model is captured well by the asymptotic result.

%%%%%%%%%%%%%%%%%%
\section*{Acknowledgements}
We are grateful to anonymous referees for their helpful comments.
Lingjiong Zhu is partially supported by the grants NSF DMS-2053454, NSF DMS-2208303.

%%%%%%%%%%%%%%%%%%%%%%%%
\bibliographystyle{plain}
\bibliography{ShortMaturityAsian}

\appendix

%%%%%%%%%%%%%%%%%%%%%%%%%%%%%%%

\section{Background on Large Deviations Theory}\label{sec:LDP}

We give in this Appendix a few basic concepts of large deviations theory from probability theory
which are in the proofs.
We refer to Dembo and Zeitouni (1998) \cite{Dembo1998} and Varadhan (1984) \cite{VaradhanLD} for more details on large deviations and its applications.

\begin{definition}[Large Deviation Principle]
A sequence $(P_\epsilon)_{\epsilon \in \mathbb{R}^+}$ of probability measures
on a topological space $X$ satisfies the large deviation principle with rate function $I: X \to \mathbb{R}$
if $I$ is non-negative, lower semicontinuous and for any measurable set $A$, we have
\begin{equation}
- \inf_{x\in A^o} I(x) \leq \liminf_{\epsilon\to 0} \epsilon \log P_\epsilon(A) \leq
\limsup_{\epsilon\to 0} \epsilon \log P_\epsilon(A) \leq - \inf_{x\in \bar A} I(x) \,,
\end{equation}
where $A^o$ denotes the interior of $A$ and $\bar A$ its closure.
\end{definition}

\begin{theorem}[Contraction Principle, see e.g. Theorem 4.2.1. in \cite{Dembo1998}]\label{Contraction:Thm}
If $F:X\rightarrow Y$ is a continuous map and 
$P_{\epsilon}$ satisfies a large deviation principle on $X$ with the rate 
function $I(x)$,
then the probability measures $Q_{\epsilon}:=P_{\epsilon}F^{-1}$ satisfies
a large deviation principle on $Y$ with the rate function
$J(y)=\inf_{x: F(x)=y}I(x)$.
\end{theorem}

%%%%%%%%%%%%%%%%%%%%%%

\section{Floating-Strike Asian Options}\label{sec:floating}

%There are many variations of the standard Asian options in the finance literature
%and one of the most used is the so-called floating-strike Asian options \cite{Levy,Ritchken,Alziary,RogersShi,HW}
%(in contrast to the fixed-strike Asian options defined in \eqref{asian:price:defn}), 

An alternative variant of Asian options traded in the markets are the so-called 
\textit{floating-strike Asian options}. In contrast to the \textit{fixed-strike Asian options} with payoffs defined as in \eqref{payoff}, these options pay $(\theta(\kappa S_T - \frac{1}{T} \int_0^T S_t dt))^+$, 
where $\kappa>0$ is the strike and $\theta = +1(-1)$ for a call (put) option.
Thus, the prices of floating-strike Asian options are given by
\begin{align}
&C_f(T):=e^{-rT}\mathbb{E}\left[\left(\kappa S_{T}-\frac{1}{T}\int_{0}^{T}S_{t}dt\right)^{+}\right],
\\
&P_f(T):=e^{-rT}\mathbb{E}\left[\left(\frac{1}{T}\int_{0}^{T}S_{t}dt-\kappa S_{T}\right)^{+}\right],
\end{align} 
where $S_{t}$ is the asset price that satisfies \eqref{eqn:S}.

The floating-strike Asian option is more difficult to price than the fixed-strike case
because the joint law of $S_{T}$ and $\frac{1}{T}\int_{0}^{T}S_{t}dt$ is needed.
Alziary et al. (1997) \cite{Alziary} give relations among the two types of Asian options and European options. 
For the particular case of the Black-Scholes model, a simple equivalence relation was
proved by Henderson and Wojakowski (2002) 
\cite{HW} among fixed-strike and floating-strike Asian options.

When $\kappa<1$, the call option is OTM, the put option is ITM;
when $\kappa>1$, the call option is ITM, the put option is OTM;
when $\kappa=1$, the call/put options are ATM. We are interested
in the short-maturity (i.e. $T\rightarrow 0$) asymptotics of these options.
We give next the leading short-maturity asymptotics for floating-strike Asian options in the local-stochastic volatility model under the technical assumptions in Section \ref{sec:model}.

We consider first the case of OTM floating-strike Asian options. 

\begin{theorem}\label{thm:OTM:floating}
Suppose Assumptions~\ref{assump:bounded}, \ref{assump:LDP} and \ref{assump:S:T:p} hold.
When $\kappa<1$, the floating-strike Asian call options are OTM and we have
\begin{equation}
\lim_{T\rightarrow 0}T\log C_{f}(T)=-\mathcal{F}_{\rho}(S_{0},V_{0},\kappa),
\end{equation}
and when $\kappa>1$, 
the floating-strike Asian put options are OTM and we have
\begin{equation}
\lim_{T\rightarrow 0}T\log P_{f}(T)=-\mathcal{F}_{\rho}(S_{0},V_{0},\kappa),
\end{equation}
where
\begin{align}
\mathcal{F}_{\rho}(S_{0},V_{0},\kappa)
:=\inf_{\substack{g(0)=\log S_{0},h(0)=\log V_{0}\\
\int_{0}^{1}e^{g(t)}dt=\kappa e^{g(1)},g,h\in\mathcal{AC}[0,1]}}\Lambda_{\rho}[g,h],\label{I:rho:floating}
\end{align}
where $\Lambda_{\rho}[g,h]$ is defined in \eqref{I:rho:2}.
\end{theorem}

%%%%%%%%%%%%%%%%%%%%%%%%%%%%%%%%%%%%%%%%%%
Next we consider the case of ATM floating-strike Asian options.

\begin{theorem}\label{thm:ATM:floating}
Suppose Assumption~\ref{assump:bounded} holds.
Further assume that $\eta(\cdot)$ and $\sigma(\cdot)$ are Lipschitz and
there exists some $C'\in(0,\infty)$ 
such that $\max_{0\leq t\leq T}\mathbb{E}[(S_{t})^{4}]\leq C'$
for any sufficiently small $T>0$.
When $\kappa=1$,
\begin{equation}
\lim_{T\rightarrow 0}\frac{C_{f}(T)}{\sqrt{T}}
=\lim_{T\rightarrow 0}\frac{P_{f}(T)}{\sqrt{T}}
=\frac{S_{0}\eta(S_{0})\sqrt{V_{0}}}{\sqrt{6\pi}}.
\end{equation}
\end{theorem}

We observe that the leading-order asymptotics in Theorem~\ref{thm:ATM:floating}
coincides with that in the fixed-strike case in Theorem~\ref{thm:ATM}.

Finally, similar to the results in Sections \ref{sec:3.2} and \ref{sec:3.3} for the fixed-strike 
Asian options \eqref{asian:price:defn},
one can obtain the expansion of the rate function \eqref{I:rho:floating} in powers of
$x:=\log(\kappa)$, and study the special cases of $\rho=\pm 1$
and the limiting case of the local volatility model as in Section~\ref{sec:main}. 
For the sake of simplicity, we omit the details.

%%%%%%%%%%%%%%%%%%%%%%
\section{Technical Proofs}\label{sec:proofs}

\begin{proof}[Proof of Theorem~\ref{thm:OTM}]
One can first show that when $K>S_{0}$
\begin{equation}\label{limT}
\lim_{T\rightarrow 0}T\log C(T)=\lim_{T\rightarrow 0}T\log\mathbb{Q}\left(\frac{1}{T}\int_{0}^{T}S_{s}ds\geq K\right).
\end{equation}
The proof for the lower bound for \eqref{limT} is quite standard and is omitted here.
Let us prove the upper bound for \eqref{limT}. 
Our following argument is inspired by the derivations for the European-style options 
that can be found in \cite{Friz2018}.
For any $U>K>S_{0}$, by applying H\"{o}lder's inequality, we have
\begin{align}
&\mathbb{E}\left[\left(\frac{1}{T}\int_{0}^{T}S_{t}dt - K\right)^+\right]  
\nonumber
\\
&= \mathbb{E}\left[\left(\frac{1}{T}\int_{0}^{T}S_{t}dt - K\right) 1_{\frac{1}{T}\int_{0}^{T}S_{t}dt\in (K, U)}\right] 
+ \mathbb{E}\left[\left(\frac{1}{T}\int_{0}^{T}S_{t}dt - K\right) 1_{\frac{1}{T}\int_{0}^{T}S_{t}dt \geq U}\right] 
\nonumber
\\
&\leq
(U - K) \mathbb{Q}\left(\frac{1}{T}\int_{0}^{T}S_{t}dt \in (K, U)\right) + \left(\mathbb{E}\left[\left(\frac{1}{T}\int_{0}^{T}S_{t}dt\right)^p \right]\right)^{1/p} \left(\mathbb{E}\left[1_{\frac{1}{T}\int_{0}^{T}S_{t}dt > U}\right]\right)^{1/q}
\nonumber
\\
&\leq
(U - K) \mathbb{Q}\left(\frac{1}{T}\int_{0}^{T}S_{t}dt\geq K\right) 
\nonumber
\\
&\qquad\qquad\qquad+ \left(\mathbb{E}\left[\left(\frac{1}{T}\int_{0}^{T}S_{t}dt\right)^p \right]\right)^{1/p} \left(\mathbb{Q}\left(\frac{1}{T}\int_{0}^{T}S_{t}dt \geq U\right)\right)^{1/q},
\label{take:log}
\end{align}
for any $p,q>1$ such that $\frac{1}{p}+\frac{1}{q}=1$,
where $p$ is chosen 
such that by Jensen's inequality, 
\begin{equation}
\mathbb{E}\left[\left(\frac{1}{T}\int_{0}^{T}S_{t}dt \right)^p\right]
\leq
\frac{1}{T}\int_{0}^{T}\mathbb{E}\left[\left(S_{t}\right)^{p}\right]dt
=O(1),
\end{equation}
as $T\rightarrow 0$ under Assumption~\ref{assump:S:T:p}.

By taking the logarithm in \eqref{take:log} and multiplying with $T$ and letting $T\rightarrow 0$, we obtain
\begin{align}
&\limsup_{T\rightarrow 0}T\log 
\mathbb{E}\left[\left(\frac{1}{T}\int_{0}^{T}S_{t}dt - K\right)^+\right]  
\nonumber
\\
&\leq
\max\left\{\limsup_{T\rightarrow 0}T\log\mathbb{Q}\left(\frac{1}{T}\int_{0}^{T}S_{t}dt\geq K\right),\frac{1}{q}\limsup_{T\rightarrow 0}T\log\mathbb{Q}\left(\frac{1}{T}\int_{0}^{T}S_{t}dt\geq U\right)\right\}.\label{eqn:K:U}
\end{align}

Next, let us show that the limits
\begin{equation*}
\lim_{T\rightarrow 0}T\log\mathbb{Q}\left(\frac{1}{T}\int_{0}^{T}S_{t}dt\geq K\right)\qquad\text{and}\qquad
\lim_{T\rightarrow 0}T\log\mathbb{Q}\left(\frac{1}{T}\int_{0}^{T}S_{t}dt\geq U\right)
\end{equation*}
exist.

Note that $(S_{t},V_{t})$ satisfies \eqref{eqn:S}-\eqref{eqn:V} which is equivalent to
\begin{align}
&\frac{dS_{t}}{S_{t}}=(r-q)dt+\eta(S_{t})\sqrt{V_{t}}\rho dZ_{t}+\eta(S_{t})\sqrt{V_{t}}\sqrt{1-\rho^{2}}dW_{t},
\\
&\frac{dV_{t}}{V_{t}}=\mu(V_{t})dt+\sigma(V_{t})dZ_{t},
\end{align}
where $W_{t},Z_{t}$ are two independent standard Brownian motions. 

Under Assumptions~\ref{assump:bounded} and \ref{assump:LDP}, by the sample-path large deviations for 
small time diffusions (see for example \cite{Varadhan} and \cite{Robertson2010}),
one can see that $\mathbb{Q}(\{(\log S_{tT},\log V_{tT}),0\leq t\leq 1\}\in\cdot)$
satisfies a sample-path large deviation principle on $L_{\infty}[0,1]$ with the rate function:
\begin{equation}\label{rate:function:LDP}
\frac{1}{2(1-\rho^{2})}\int_{0}^{1}\left(\frac{g'(t)}{\eta(e^{g(t)})\sqrt{e^{h(t)}}}-\frac{\rho h'(t)}{\sigma(e^{h(t)})}\right)^{2}dt
+\frac{1}{2}\int_{0}^{1}\left(\frac{h'(t)}{\sigma(e^{h(t)})}\right)^{2}dt,    
\end{equation}
with $g(0)=\log S_{0}$, $h(0)=\log V_{0}$ and $g,h\in\mathcal{AC}[0,1]$, where $\mathcal{AC}[0,1]$ is the space of absolutely continuous functions on $[0,1]$, and the rate function is $+\infty$ otherwise.

By an application of the contraction principle (see for example Theorem 4.2.1. in \cite{Dembo1998}, restated in Theorem~\ref{Contraction:Thm}), one can compute that
\begin{align}
&\lim_{T\rightarrow 0}T\log\mathbb{Q}\left(\frac{1}{T}\int_{0}^{T}S_{s}ds\geq K\right)
\nonumber
\\
&=-\inf_{\substack{g(0)=\log S_{0},h(0)=\log V_{0}\\
\int_{0}^{1}e^{g(t)}dt=K,g,h\in\mathcal{AC}[0,1]}}
\Bigg\{\frac{1}{2(1-\rho^{2})}\int_{0}^{1}\left(\frac{g'(t)}{\eta(e^{g(t)})\sqrt{e^{h(t)}}}-\frac{\rho h'(t)}{\sigma(e^{h(t)})}\right)^{2}dt
\nonumber
\\
&\qquad\qquad\qquad\qquad\qquad\qquad\qquad\qquad\qquad
+\frac{1}{2}\int_{0}^{1}\left(\frac{h'(t)}{\sigma(e^{h(t)})}\right)^{2}dt\Bigg\}.
\end{align}

The proof for the put options is similar and hence omitted here.
\end{proof}

%%%%%%%%%%%%%%%%%%%%%%%%%%%%%%%%%%%%

\begin{proof}[Proof of Proposition~\ref{prop:first:order}]

We start with an expansion for the functions $g,h$ in powers of log-moneyness $x=\log(K/S_0)$ of the form
\begin{align}
&g(t)=g_{0}(t)+xg_{1}(t)+x^{2}g_{2}(t)+x^{3}g_{3}(t)+O(x^{4}),\nonumber
\\
&h(t)=h_{0}(t)+xh_{1}(t)+x^{2}h_{2}(t)+x^{3}h_{3}(t)+O(x^{4}),\label{g:h:expansion}
\end{align}
We expand also the Lagrange multiplier as $\lambda=\lambda_{0}+x\lambda_{1}+x^{2}\lambda_{2}+x^{3}\lambda_{3}+O(x^{4})$ as $x\rightarrow 0$.

The zero-th order terms in these expansions are
%where we notice that
$g_{0}(t)\equiv\log S_{0}$, $h_{0}(t)\equiv\log V_{0}$
such that $g'_{0}(t)\equiv 0$ and $h'_{0}(t)\equiv 0$.
The transversality conditions $g'(1) = h'(1) = 0$ are satisfied if and only if we have 
$g'_k(1) = h'_k(1) = 0$ for all $k\geq 1$. Also, the boundary conditions $g(0) = \log S_0, h(0)=\log V_0$ imply that one must have $g_k(0) = h_k(0)=0$ for all $k\geq 1$.

%Moreover, one can check that $\lambda_{0}=0$
%by plugging $g_{0}(t)\equiv\log S_{0}$, $h_{0}(t)\equiv\log V_{0}$
%into \eqref{EL:1}-\eqref{EL:2}.

The constraint $\int_0^1 e^{g(t)}dt = K$ relates $g_k(t)$ to all $g_j(t)$ of lower order $0 \leq j < k$.
This is written equivalently as
\begin{equation}
S_0 \int_0^1 e^{x g_1(t) + x^2 g_2(t) + O(x^{3}) }dt = S_0 e^x, 
\end{equation}
as $x\rightarrow 0$.
Expanding in $x$ and selecting terms of the same power of $x$ on both sides gives the constraints
\begin{align}\label{norm}
& \int_0^1 g_1(t) dt = 1 \,,\\
& \int_0^1 \left(g_2(t) + \frac12 (g_1(t))^2 \right)dt = \frac12 \,,\nonumber \\
& \int_0^1 \left(g_3(t) + g_1(t) g_2(t) + \frac16 (g_1(t))^3 \right)dt = \frac16\,,\nonumber\\
&\cdots\cdots\cdots\cdots \nonumber
\end{align}

%By plugging 
We substitute the expansions \eqref{g:h:expansion} into the Euler-Lagrange equations 
\eqref{EL:1}-\eqref{EL:2} and expand in $x$. Let us consider the terms of given order in $x$ resulting from this expansion.

\textbf{Order $O(x^0)$.}
At order $O(x^0)$, the equation~\eqref{EL:1} gives
$\lambda_0 S_0 = 0$ which gives $\lambda_0=0$.
Both sides of the equation~\eqref{EL:2}  vanish identically at this order.

\textbf{Order $O(x)$.}
%$g_{0}(t)\equiv\log S_{0}$, $h_{0}(t)\equiv\log V_{0}$, $g'_{0}(t)\equiv 0$ and $h'_{0}(t)\equiv 0$, 
At order $O(x)$, the two equations become
%and taking the coefficients in the $x$ term
%in the expansion to be zero, we obtain that
\begin{align}
\frac{d}{dt}\left(\frac{1}{1-\rho^{2}}\frac{g'_{1}(t)}{\eta^{2}_0 V_{0}}
-\frac{\rho}{1-\rho^{2}}\frac{h'_{1}(t)}{\eta_0 \sqrt{V_{0}}\sigma_0}\right)
=\lambda_{1}S_{0},\label{EL:1:first}
\end{align}
and
\begin{align}
\frac{d}{dt}\left(\frac{1}{1-\rho^{2}}\frac{h'_{1}(t)}{\sigma^{2}_0}
-\frac{\rho}{1-\rho^{2}}\frac{g'_{1}(t)}{\eta_0 \sqrt{V_{0}}\sigma_0 }\right)=0,\label{EL:2:first}
\end{align}
with the constraints that $g_{1}(0)=0$, $h_{1}(0)=0$ and $\int_{0}^{1}g_{1}(t)dt=1$,
and the transversality condition gives $g_{1}'(1)=h_{1}'(1)=0$, where $\eta_{0},\sigma_{0}$ are defined in \eqref{eta:expansion} and \eqref{sigma:expansion}.
We can re-write \eqref{EL:1:first}-\eqref{EL:2:first} as
\begin{align}
&\frac{g''_{1}(t)}{\eta^{2}_0V_{0}}-\frac{\rho h''_{1}(t)}{\eta_0\sqrt{V_{0}}\sigma_0}=\lambda_{1}S_{0}(1-\rho^{2}),
\\
&\frac{h''_{1}(t)}{\sigma^{2}_0}=\frac{\rho g''_{1}(t)}{\eta_0\sqrt{V_{0}}\sigma_0},
\end{align}
which implies that
\begin{align}
&g''_{1}(t)=\lambda_{1}S_{0}\eta^{2}_0 V_{0},\qquad\qquad \,\, g_{1}(0)=g'_{1}(1)=0,
\\
&h''_{1}(t)=\lambda_{1}S_{0}\rho\sigma_0\eta_0\sqrt{V_{0}},\qquad h_{1}(0)=h'_{1}(1)=0,
\end{align}
with $\int_{0}^{1}g_{1}(t)dt=1$. 
These equations can be integrated using the boundary condition $g'_1(1)=h'(1)=0$ to give
\begin{align}
&g'_{1}(t)= \lambda_{1}S_{0}\eta^{2}(S_{0})V_{0}(t-1),
\\
&h'_{1}(t)= \lambda_{1}S_{0}\rho\sigma(V_{0})\eta(S_{0})\sqrt{V_{0}}(t-1) \,.
\end{align}
Integrating again using the boundary conditions $g_1(0)=h_1(0)=0$ gives
%One can solve for $g_{1},h_{1}$ to obtain:
\begin{align}
&g_{1}(t)=\frac{1}{2}\lambda_{1}S_{0}\eta^{2}_0 V_{0}(t^{2}-2t),
\\
&h_{1}(t)=\frac{1}{2}\lambda_{1}S_{0}\rho\sigma_0 \eta_0\sqrt{V_{0}}(t^{2}-2t),
\end{align}

The constant $\lambda_1$ is determined from the first normalization condition (\ref{norm})
\begin{equation}
\int_{0}^{1}g_{1}(t)dt=-\frac{1}{2}\lambda_{1}S_{0}\eta^{2}_0 V_{0}\frac{2}{3}=1,
\end{equation}
which gives $\lambda_{1}=-\frac{3}{S_{0}\eta^{2}_0V_{0}}$. We conclude that
\begin{align}
&g_{1}(t)=\frac{3}{2}(2t-t^{2}),
\\
&h_{1}(t)=\frac{3\rho\sigma_0}{2\eta_0 \sqrt{V_{0}}}(2t-t^{2}).
\end{align}

Finally, by plugging \eqref{g:h:expansion} into \eqref{I:rho:2}, it follows from \eqref{I:rho:1}
that
\begin{align}
&\mathcal{I}_{\rho}(S_{0},V_{0},S_{0}e^{x})
\nonumber
\\
&=\frac{x^{2}}{2(1-\rho^{2})}\int_{0}^{1}\left(\frac{g'_{1}(t)}{\eta(S_{0})\sqrt{V_{0}}}-\frac{\rho h'_{1}(t)}{\sigma_0}\right)^{2}dt
+\frac{x^{2}}{2}\int_{0}^{1}\left(\frac{h'_{1}(t)}{\sigma_0}\right)^{2}dt
+O(x^{3})
\nonumber
\\
&=\frac{3}{2\eta^{2}_0V_{0}}x^{2}+O(x^{3}),
\end{align}
as $x\rightarrow 0$.
This completes the proof for the first term in \eqref{JAexp}.

\textbf{Order $O(x^2)$.} The proof is completely analogous to that of the first term, so we give only the main steps. 
At $O(x^2)$ the two Euler-Lagrange equations are again two linear equations in $g''_2(t), h''_2(t)$, which are easily solved. The results are quadratic polynomials in $t$, which we write as
\begin{equation}\label{g2pp}
g''_2(t) = a + b t + c t^2\,,\quad h''_2(t) = \bar a + \bar b t + \bar c t^2\,.
\end{equation}

The coefficients of these polynomials are 
\begin{align}
& a = \frac{1}{2\eta_0} \left(18 \eta_1 + 9 \frac{\rho \sigma_0}{\sqrt{V_0}} \right) + \eta_0^2 \lambda_2 S_0 V_0, \\
%\frac{39}{10} + \frac{9}{10\eta_0} \left(16\eta_1 + 8 \frac{\rho \sigma_0}{\sqrt{V_0}} \right), \\ 
& b = - 9 - \frac{9}{\eta_0} \left(4 \eta_1 + 2 \frac{\rho \sigma_0}{\sqrt{V_0}} \right), \\ 
& c = \frac{9}{2} + \frac{1}{\eta_0} \left(18\eta_1 + 9 \frac{\rho \sigma_0}{\sqrt{V_0}} \right), 
\end{align}
and
\begin{align}
& \bar a = \frac{\sigma_0}{2\eta_0^2 V_0}\left( - 9\left(1-\rho^2\right) \sigma_0 + 
2\rho\left(9\rho \sigma_1 + \eta_0^3 \lambda_2 S_0 V_0^{3/2}\right) \right), \\
%\frac{3\sigma}{10\eta_0^2 V_0} \left(3 (-5+8\rho^2) \sigma_0 + 
%\rho \left( 30 \rho \sigma_1 + (13 \eta_0 + 18\eta_1) \sqrt{V_0} \right)\right),\\
& \bar b =  -\frac{9\sigma}{2\eta_0^2 V_0} \left( \left(-2+3\rho^2\right) \sigma_0 + 
2\rho \left( 3 \rho \sigma_1 + ( \eta_0 + \eta_1) \sqrt{V_0} \right)\right), \\ 
& \bar c = \frac{9\sigma}{4\eta_0^2 V_0} \left( (-2+3\rho^2) \sigma_0 + 
2\rho \left( 3 \rho \sigma_1 + ( \eta_0 + \eta_1) \sqrt{V_0} \right)\right).
\end{align}

The equations in (\ref{g2pp}) can be integrated using the boundary conditions at $t=0$ and $t=1$
as before to give
\begin{equation}
g'_2(t) = (t-1) \left(a + \frac12 b (t+1) + \frac13 c (t^2+t+1) \right), 
\end{equation}
and
\begin{equation}\label{g2sol}
g_2(t) = \frac12 a t (t-2) + \frac16 b t (t^2-3) + \frac{1}{12} c t (t^3-4),
\end{equation}
and analogous for $h_2(t)$ with the replacements $a \to \bar a, b\to \bar b, c\to \bar c$.

The normalization condition for $g_2(t)$ reads $\int_0^1 g_2(t) = \frac12 - \frac12 \int_0^1 (g_1(t))^2 dt = -\frac{1}{10}$ where we substituted the result $g_1(t) = \frac32 (2t-t^2)$ found above. 
Substituting here the expression (\ref{g2sol}), this gives a constraint for the constants $a,b,c$
\begin{equation}
\frac13 a + \frac{5}{24} b + \frac{3}{20} c = \frac{1}{10}\,.
\end{equation}
From this we determine the constant 
\begin{equation}
\lambda_2 = \frac{3}{10 \eta_0^3 S_0 V_0^{3/2}} \left(9 \rho \sigma_0 + 13 \eta_0 \sqrt{V_0} + 18 \eta_1 \sqrt{V_0}\right).
\end{equation}

Substituting into $a$ and $\bar a$ we get the final expressions for these coefficients
\begin{align}
& a = \frac{39}{10} + \frac{9}{10\eta_0} \left(16\eta_1 + 8 \frac{\rho \sigma_0}{\sqrt{V_0}} \right), \\ 
& \bar a = \frac{3\sigma}{10\eta_0^2 V_0} \left(3 (-5+8\rho^2) \sigma_0 + 
\rho \left( 30 \rho \sigma_1 + (13 \eta_0 + 18\eta_1) \sqrt{V_0} \right)\right) \,.
\end{align}
Finally, substituting the expressions for $g'_2(t), h'_2(t)$ into the expansion of the rate function to order $O(x^3)$ gives the stated result. 

\textbf{Order $O(x^3)$.} At this order we have the unknowns $g_3(t), h_3(t), \lambda_3$. 
The Euler-Lagrange equations determine the second derivatives $g''_3(t), h''_3(t)$ in terms of lower order functions. The solutions for these second derivatives are quintic polynomials in $t$, 
which we denote as
\begin{equation}
g''_3(t) = c_0 + c_1 t + c_2 t^2 + \cdots + c_5 t^5\,,
\end{equation}
and analogous for $h''_3(t)$ in terms of $\bar c_k$. The Euler-Lagrange equations determine the 12 coefficients
$c_k, \bar c_k$ with $k=0,1,2,\cdots , 5$. From $g''_3(t)$ we find $g'_3(t)$ and $g_3(t)$ by integration, using the boundary conditions $g'_3(1)=0, g_3(0)=0$. The results are
\begin{equation}
g'_3(t) = (t-1)\left(c_0 + \frac12 c_1(t+1) + \frac13 c_2 \left(t^2+t+1\right) + \cdots + \frac16 c_5 \left(t^4+t^3+t^2+t+1\right)\right)\,,
\end{equation}
and
\begin{equation}\label{g3sol}
g_3(t) = \frac12 c_0 t (t-2) + \frac{1}{2\cdot 3} c_1 t (t^2-3) + \frac{1}{3\cdot 4} c_2 t(t^3-4) + \cdots + 
\frac{1}{6\cdot 7} c_5 t (t^6-7)\,.
\end{equation}
Similar expressions hold for $h'_3(t)$ and $h_3(t)$, with the replacements $c_k \to \bar c_k$, respectively.

The coefficients $c_0, \bar c_0$ depend, in addition to $\sigma_k,\eta_k, V_0,\rho$, also on 
$\lambda_3$. 
On the other hand, $c_k, \bar c_k$ with $1\leq k \leq 5$ do not depend on $\lambda_3$.
In order to fix $\lambda_3$, we need to use the normalization equation (\ref{norm}). The constraint for $g_3(t)$ is obtained from the third line in (\ref{norm}) and reads explicitly
\begin{equation}\label{g3const}
\int_0^1 g_3(t) dt = \frac16 - \int_0^1 g_1(t) g_2(t) dt - \frac16 \int_0^1 (g_1(t))^3 dt =
\frac{1}{1050\eta_0} \left( 13 \eta_0 - 72 \eta_1 - 36 \frac{\rho\sigma_0}{\sqrt{V_0}}\right)\,.
\end{equation}
The integral on the left-hand side can be evaluated using (\ref{g3sol}). 
\begin{equation}
\int_0^1 g_3(t) dt = -\frac13 c_0 - \frac{5}{24} c_1 - \frac{3}{20} c_2 
         - \frac{7}{60}  c_3 - \frac{2}{21} c_4 - \frac{9}{112} c_5 \,.
\end{equation}
Combined with (\ref{g3const}), this gives a result for $c_0$ and thus also $\lambda_3$ as follows:
\begin{align}\label{lam3sol}
\lambda_3 &= \frac{1}{44800 \eta_0^4 S_0 V_0^2} 
\Big(-15750 \sigma_0^2 + 15669 \rho^2 \sigma_0^2\\
& \qquad\qquad\qquad\qquad\quad- 17370 \rho^2 \sigma_0 \sigma_1 - 135198 \eta_0 \rho \sigma_0 \sqrt{V_0} - 
   17694 \eta_1 \rho \sigma_0 \sqrt{V_0} 
   \nonumber
   \\
   &\qquad\qquad\qquad\qquad\qquad- 121472 \eta_0^2 V_0 - 
   270396 \eta_0 \eta_1 V_0 - 47484 \eta_1^2 V_0 + 24840 \eta_0 \eta_2 V_0\Big) \,.\nonumber
\end{align}
The coefficients $c_k, \bar c_k$ are given in Appendix \ref{sec:appC} 
as their expressions are rather lengthy. 

The $O(x^4)$ term in the rate function is found by integrating over $t$ the coefficient of this order in the expansion of the integrand in $\Lambda_\rho[g,h]$, see (\ref{I:rho:2}).
This completes the proof.
\end{proof}

%%%%%%%%%%%%%%%%%%%%%%%%%%%%%%%%%%%

\begin{proof}[Proof of Proposition~\ref{prop:pm:1}]
As $\rho\rightarrow\pm 1$, we must have
$\frac{g'(t)}{\eta(e^{g(t)})\sqrt{e^{h(t)}}}-\frac{\rho h'(t)}{\sigma(e^{h(t)})}\rightarrow 0$, 
$0\leq t\leq 1$; otherwise $\Lambda_{\rho}[g,h]$ would approach to $+\infty$.
Therefore, given $h$, when $\rho=\pm 1$, the optimal $g$ satisfies
\begin{equation}
\frac{g'(t)}{\eta(e^{g(t)})\sqrt{e^{h(t)}}}=\frac{\pm h'(t)}{\sigma(e^{h(t)})},\qquad 0\leq t\leq 1,
\end{equation}
which implies that
\begin{equation}
\int_{0}^{t}\frac{g'(s)}{\eta(e^{g(s)})}ds=\int_{0}^{t}\frac{\pm \sqrt{e^{h(s)}}h'(s)}{\sigma(e^{h(s)})}ds,\qquad 0\leq t\leq 1,
\end{equation}
which is equivalent to
\begin{equation}
\int_{\log S_{0}}^{g(t)}\frac{dx}{\eta(e^{x})}=\int_{\log V_{0}}^{h(t)}\frac{\pm \sqrt{e^{x}}dx}{\sigma(e^{x})},\qquad 0\leq t\leq 1,
\end{equation}
where we used the constraints $g(0)=\log S_{0}$ and $h(0)=\log V_{0}$.
We can further compute that this is equivalent to
\begin{equation}
\int_{\log S_{0}}^{g(t)}\frac{e^{x}dx}{e^{x}\eta(e^{x})}=\int_{\log V_{0}}^{h(t)}\frac{\pm e^{x}dx}{\sqrt{e^{x}}\sigma(e^{x})},\qquad 0\leq t\leq 1,
\end{equation}
which is equivalent to
\begin{equation}
\int_{S_{0}}^{e^{g(t)}}\frac{dx}{x\eta(x)}=\int_{V_{0}}^{e^{h(t)}}\frac{\pm dx}{\sqrt{x}\sigma(x)},\qquad 0\leq t\leq 1.
\end{equation}
Therefore, given $h$, the optimal $g$ is given by
\begin{equation}
e^{g(t)}=\mathcal{F}_{\pm}(e^{h(t)}),\qquad 0\leq t\leq 1,
\end{equation}
where $\mathcal{F}_{\pm}(\cdot)$ is defined as:
\begin{equation}
\int_{S_{0}}^{\mathcal{F}_{\pm}(x)}\frac{dy}{y\eta(y)}=\int_{V_{0}}^{x}\frac{\pm dy}{\sqrt{y}\sigma(y)},
\end{equation}
for any $x>0$.
Hence, we conclude that
\begin{equation}
\mathcal{I}_{\pm 1}(S_{0},V_{0},K)
=\inf_{\substack{h(0)=\log V_{0}\\
\int_{0}^{1}\mathcal{F}_{\pm}(e^{h(t)})dt=K, h\in\mathcal{AC}[0,1]}}\frac{1}{2}\int_{0}^{1}\left(\frac{h'(t)}{\sigma(e^{h(t)})}\right)^{2}dt.
\end{equation}
This completes the proof.
\end{proof}

%%%%%%%%%%%%%%%%%%%%%%%%%%%%%%%%%%%

\begin{proof}[Proof of Theorem~\ref{thm:ATM}]
We only provide the proof for the ATM call option. 
The case for the put option is similar.
First of all, it is easy to see that
\begin{equation}
\left|C(T)-\mathbb{E}\left[\left(\frac{1}{T}\int_{0}^{T}S_{t}dt-S_{0}\right)^{+}\right]\right|
=O(T),
\end{equation}
as $T\rightarrow 0$, which follows from the estimate:
\begin{align}
\left|C(T)-\mathbb{E}\left[\left(\frac{1}{T}\int_{0}^{T}S_{t}dt-S_{0}\right)^{+}\right]\right|
&\leq
\left|1-e^{-rT}\right|\cdot\left|\mathbb{E}\left[\left(\frac{1}{T}\int_{0}^{T}S_{t}dt-S_{0}\right)^{+}\right]\right|
\nonumber
\\
&\leq\left|1-e^{-rT}\right|\cdot\left(\frac{1}{T}\int_{0}^{T}\mathbb{E}[S_{t}]dt+S_{0}\right)
\nonumber
\\
&=\left|1-e^{-rT}\right|\cdot\left(\frac{1}{T}\int_{0}^{T}S_{0}e^{(r-q)t}dt+S_{0}\right).
\end{align}

For the ATM case (i.e. $K=S_{0}$), it is proved in \cite{PWZ2024} that
for the local-stochastic volatility model \eqref{eqn:S}-\eqref{eqn:V}, 
under Assumption~\ref{assump:bounded}, and the assumption that $\eta(\cdot),\sigma(\cdot)$ are Lipschitz and there exists some $C'\in(0,\infty)$ 
such that $\max_{0\leq t\leq T}\mathbb{E}[(S_{t})^{4}]\leq C'$
for any sufficiently small $T>0$, 
we have that there exists some $C>0$, 
such that for any sufficiently small $T>0$, 
\begin{equation}
\mathbb{E}\left|S_{t}-\hat{S}_{t}\right|^{2}
\leq CT^{3/2},
\end{equation}
for any $0\leq t\leq T$, where
\begin{equation}
\hat{S}_{t}:=S_{0}+S_{0}\eta(S_{0})\sqrt{V_{0}}B_{t},
\end{equation}
with $B_{t}$ being a standard Brownian motion as in \eqref{eqn:S}.

Since the map $x\mapsto x^{+}$ is $1$-Lipschitz, we have
\begin{align}
\left|\mathbb{E}\left[\left(\frac{1}{T}\int_{0}^{T}S_{t}dt-S_{0}\right)^{+}\right]
-\mathbb{E}\left[\left(\frac{1}{T}\int_{0}^{T}\hat{S}_{t}dt-S_{0}\right)^{+}\right]\right|
&\leq
\frac{1}{T}\int_{0}^{T}\mathbb{E}\left[\left|S_{t}-\hat{S}_{t}\right|\right]dt
\nonumber
\\
&\leq
\frac{1}{T}\int_{0}^{T}\left(\mathbb{E}\left[\left|S_{t}-\hat{S}_{t}\right|^{2}\right]\right)^{1/2}dt
\nonumber
\\
&\leq\sqrt{C}T^{3/4},
\end{align}
for any sufficiently small $T$.

Finally, we can compute that
\begin{align}
\mathbb{E}\left[\left(\frac{1}{T}\int_{0}^{T}(S_{0}+S_{0}\eta(S_{0})\sqrt{V_{0}}B_{t})dt-S_{0}\right)^{+}\right]
&=S_{0}\eta(S_{0})\sqrt{V_{0}}\mathbb{E}\left[\left(\frac{1}{T}\int_{0}^{T}B_{t}dt\right)^{+}\right]
\nonumber
\\
&=S_{0}\eta(S_{0})\sqrt{V_{0}}\frac{\sqrt{T}}{\sqrt{3}}\mathbb{E}[X^{+}]
\nonumber
\\
&=S_{0}\eta(S_{0})\sqrt{V_{0}}\frac{\sqrt{T}}{\sqrt{6\pi}},
\end{align}
where $X\sim\mathcal{N}(0,1)$.
This completes the proof.
\end{proof}

%%%%%%%%%%%%%%%%%%%%%%%%%%%%%%

\begin{proof}[Proof of Theorem~\ref{thm:OTM:floating}]
The proof follows along similar lines as the proof of Theorem~\ref{thm:OTM}
by applying large deviations and a contraction principle 
and we omit the details here.
\end{proof}

%%%%%%%%%%%%%%%%%%%%%%%%%%%%%%%%%%%%%%
\begin{proof}[Proof of Theorem~\ref{thm:ATM:floating}]
We only provide the proof for the ATM call option. 
The case for the put option is similar.
By following the same arguments as in the proof of Theorem~\ref{thm:ATM}, 
we can show that
\begin{equation}
\left|C_{f}(T)-\mathbb{E}\left[\left(\hat{S}_{T}-\frac{1}{T}\int_{0}^{T}\hat{S}_{t}dt\right)^{+}\right]\right|
\leq O(T^{3/4}),
\end{equation}
as $T\rightarrow 0$, where
\begin{equation}
\hat{S}_{t}:=S_{0}+S_{0}\eta(S_{0})\sqrt{V_{0}}B_{t},
\end{equation}
with $B_{t}$ being a standard Brownian motion as in \eqref{eqn:S}.

In addition, we can compute that
\begin{align}
&\mathbb{E}\left[\left(S_{0}+S_{0}\eta(S_{0})\sqrt{V_{0}}B_{T}-\frac{1}{T}\int_{0}^{T}(S_{0}+S_{0}\eta(S_{0})\sqrt{V_{0}}B_{t})dt\right)^{+}\right]
\nonumber
\\
&=S_{0}\eta(S_{0})\sqrt{V_{0}}\mathbb{E}\left[\left(B_{T}-\frac{1}{T}\int_{0}^{T}B_{t}dt\right)^{+}\right].
\end{align}
Note that $B_{T}-\frac{1}{T}\int_{0}^{T}B_{t}dt$ is Gaussian with mean $0$ and variance
\begin{align}
\mathbb{E}\left[\left(B_{T}-\frac{1}{T}\int_{0}^{T}B_{t}dt\right)^{2}\right]
&=\mathbb{E}[B_{T}^{2}]+\frac{1}{T^{2}}\mathbb{E}\left[\left(\int_{0}^{T}B_{t}dt\right)^{2}\right]-\frac{2}{T}\int_{0}^{T}\mathbb{E}[B_{T}B_{t}]dt
\nonumber
\\
&=T+\frac{T}{3}-\frac{2}{T}\frac{T^{2}}{2}=\frac{T}{3}.
\end{align}
Hence, we conclude that
\begin{align}
&\mathbb{E}\left[\left(S_{0}+S_{0}\eta(S_{0})\sqrt{V_{0}}B_{T}-\frac{1}{T}\int_{0}^{T}(S_{0}+S_{0}\eta(S_{0})\sqrt{V_{0}}B_{t})dt\right)^{+}\right]
\nonumber
\\
&=S_{0}\eta(S_{0})\sqrt{V_{0}}\frac{\sqrt{T}}{\sqrt{3}}\mathbb{E}[X^{+}]
=S_{0}\eta(S_{0})\sqrt{V_{0}}\frac{\sqrt{T}}{\sqrt{6\pi}},
\end{align}
where $X\sim\mathcal{N}(0,1)$.
This completes the proof.
\end{proof}

%%%%%%%%%%%%%%%%%%%%%%%%%%%%%%
\section{Details for the $g_3(t), h_3(t)$ Functions}
\label{sec:appC}

We give here the 12 coefficients appearing in the functions $g_3(t)$ and $h_3(t)$ giving the optimal paths to order $O(x^3)$ for the asset price and volatility, respectively.
The function $g_3(t)$ is expressed as in (\ref{g3sol}), in terms of 6 coefficients $c_k$ which are given as follows:
\begin{align}
c_0 &=
 \frac{1}{44800 \eta_0^2 V_0}
  \Big(9 (-1750 + 30861 \rho^2) \sigma_0^2 - 
    18 \rho \sigma_0 (12165 \rho \sigma_1 + (791 \eta_0 - 46057 \eta_1) \sqrt{V_0}) \nonumber \\
    &\qquad\qquad\qquad\qquad+ 
    4 (-30368 \eta_0^2 - 7119 \eta_0 \eta_1 + 149409 \eta_1^2 + 
       6210 \eta_0 \eta_2) V_0 \Big)\,,
\end{align}
%%%%%%%%%%%%%%%%%%%
\begin{align}
c_1 &= \frac{1}{80 \eta_0^2 V_0}
\Big(360 \sigma_0^2 - 2772 \rho^2 \sigma_0^2 + 
  1800 \rho^2 \sigma_0 \sigma_1 - 648 \eta_0 \rho \sigma_0 \sqrt{V_0} \\
  &\qquad\qquad\qquad- 
  7848 \eta_1 \rho \sigma_0 \sqrt{V_0} + 1152 \eta_0^2 V_0 - 1296 \eta_0 \eta_1 V_0 - 
  3888 \eta_1^2 V_0 - 4320 \eta_0 \eta_2 V_0 \Big)\,, \nonumber
\end{align}
%%%%%%%%%%%%%%%%%%%%%%%%%
\begin{align}
c_2 &= 
\frac{1}{80 \eta_0^2 V_0}
\Big(-900 \sigma_0^2 + 5202 \rho^2 \sigma_0^2 - 
  2340 \rho^2 \sigma_0 \sigma_1 + 2556 \eta_0 \rho \sigma_0 \sqrt{V_0} \\
  &\qquad\qquad\quad+ 
  14868 \eta_1 \rho \sigma_0 \sqrt{V_0} - 2016 \eta_0^2 V_0 + 5112 \eta_0 \eta_1 V_0 + 
  4968 \eta_1^2 V_0 + 15120 \eta_0 \eta_2 V_0 \Big)\,, \nonumber 
\end{align}
%%%%%%%%%%%%%%%%%%%%%%%%%%%%%
\begin{align}
c_3 &= 
\frac{1}{80 \eta_0^2 V_0}
\Big(900 \sigma_0^2 - 4518 \rho^2 \sigma_0^2 + 
  900 \rho^2 \sigma_0 \sigma_1 - 3096 \eta_0 \rho \sigma_0 \sqrt{V_0} \\
  &\qquad\qquad\quad- 
  13572 \eta_1 \rho \sigma_0 \sqrt{V_0} + 1440 \eta_0^2 V_0 - 6192 \eta_0 \eta_1 V_0 - 
  4032 \eta_1^2 V_0 - 17280 \eta_0 \eta_2 V_0 \Big)\,, \nonumber
\end{align}
%%%%%%%%%%%%%%%%%%%%%%%%%%%%%%%%
\begin{align}
c_4 &=
\frac{1}{80 \eta_0^2 V_0}
\Big(-450 \sigma_0^2 + 1935 \rho^2 \sigma_0^2 + 
  90 \rho^2 \sigma_0 \sigma_1 + 1530 \eta_0 \rho \sigma_0 \sqrt{V_0} \\
  &\qquad\qquad\qquad+ 
  6030 \eta_1 \rho \sigma_0 \sqrt{V_0} - 360 \eta_0^2 V_0 + 3060 \eta_0 \eta_1 V_0 + 
  2340 \eta_1^2 V_0 + 7560 \eta_0 \eta_2 V_0)\,, \nonumber
\end{align}
%%%%%%%%%%%%%%%%%%%%%%%%%%%%%%%%%
\begin{align}
 c_5 &= 
 \frac{1}{80 \eta_0^2 V_0}
 \Big(90 \sigma_0^2 - 315 \rho^2 \sigma_0^2 - 90 \rho^2 \sigma_0 \sigma_1 - 
  270 \eta_0 \rho \sigma_0 \sqrt{V_0} - 990 \eta_1 \rho \sigma_0 \sqrt{V_0} \\
  &\qquad\qquad\qquad- 
  540 \eta_0 \eta_1 V_0 - 540 \eta_1^2 V_0 - 1080 \eta_0 \eta_2 V_0 \Big)\,. \nonumber
\end{align}

The function $h_3(t)$ has the same form as (\ref{g3sol}) but with different coefficients $\bar c_k$.
They are given below;
\begin{align}
\bar c_0 &=
\frac{1}{44800 \eta_0^3 V_0^{3/2}} 
\sigma_0 \Big(246330 \rho \sigma_0^2 + 9 \rho^3 \sigma_0 (1741 \sigma_0 - 1930 \sigma_1) - 
   201600 \rho \sigma_0 \sigma_1 \\
   &\qquad\qquad\qquad\qquad\quad+ 40320 (3 \eta_0 + 8 \eta_1) \sigma_0 \sqrt{V_0} - 
   18 (7511 \eta_0 - 28137 \eta_1) \rho^2 \sigma_0 \sqrt{V_0} \nonumber \\
   & \qquad\qquad\qquad\qquad\qquad+ 
   4 (-30368 \eta_0^2 - 7119 \eta_0 \eta_1 + 149409 \eta_1^2 + 
      6210 \eta_0 \eta_2) \rho V_0 \Big)\,,\nonumber
\end{align}
%%%%%%%%%%%%%%%%%%%%%%%%
\begin{align}
\bar c_1 &=
\frac{1}{80 \eta_0^3 V_0^{3/2}}
\sigma_0 \Big(-1476 \rho \sigma_0^2 - 648 \rho^3 \sigma_0^2 + 2520 \rho \sigma_0 \sigma_1 - 
   1296 \rho^3 \sigma_0 \sigma_1  \\
   & \qquad\qquad\qquad\qquad- 1152 \eta_0 \sigma_0 \sqrt{V_0}
   - 1872 \eta_1 \sigma_0 \sqrt{V_0} + 720 \eta_0 \rho^2 \sigma_0 \sqrt{V_0}  \nonumber \\
   & \qquad\qquad\qquad\qquad\qquad
   - 4824 \eta_1 \rho^2 \sigma_0 \sqrt{V_0} - 432 \eta_0 \rho^2 \sigma_1 \sqrt{V_0}
   - 1152 \eta_1 \rho^2 \sigma_1 \sqrt{V_0} 
   \nonumber
   \\
   &\qquad\qquad\qquad\qquad\qquad\qquad+ 1152 \eta_0^2 \rho V_0 
   - 864 \eta_0 \eta_1 \rho V_0 - 2736 \eta_1^2 \rho V_0 - 4320 \eta_0 \eta_2 \rho V_0 \Big)\,,\nonumber
\end{align}
%%%%%%%%%%%%%%%%%%%%%%%%%%
\begin{align}
\bar c_2 &=
\frac{1}{80 \eta_0^3 V_0^{ 3/2}}
\sigma_0 \Big(1818 \rho \sigma_0^2 + 1512 \rho^3 \sigma_0^2 - 5580 \rho \sigma_0 \sigma_1 + 
   4104 \rho^3 \sigma_0 \sigma_1 + 1080 \rho^3 \sigma_1^2 \\
   &\qquad\qquad\qquad\qquad + 2160 \rho^3 \sigma_0 \sigma_2 + 
   2016 \eta_0 \sigma_0 \sqrt{V_0} + 2376 \eta_1 \sigma_0 \sqrt{V_0} - 
   324 \eta_0 \rho^2 \sigma_0 \sqrt{V_0}  \nonumber \\
   &\qquad\qquad\qquad\qquad\quad+ 8064 \eta_1 \rho^2 \sigma_0 \sqrt{V_0} + 
   1728 \eta_0 \rho^2 \sigma_1 \sqrt{V_0} + 2808 \eta_1 \rho^2 \sigma_1 \sqrt{V_0} 
   \nonumber
   \\
   &\qquad\qquad\qquad\qquad\quad\quad- 
   2016 \eta_0^2 \rho V_0 + 3384 \eta_0 \eta_1 \rho V_0 + 1080 \eta_1^2 \rho V_0 + 
   12960 \eta_0 \eta_2 \rho V_0 \Big)\,,\nonumber
\end{align}
%%%%%%%%%%%%%%%%%%%%%%%%
\begin{align}
\bar c_3 &=
\frac{1}{80 \eta_0^3 V_0^{ 3/2}}
\sigma_0 \Big(-990 \rho \sigma_0^2 - 1404 \rho^3 \sigma_0^2 + 5580 \rho \sigma_0 \sigma_1 - 
   4968 \rho^3 \sigma_0 \sigma_1 \\
   &\qquad\qquad\qquad\qquad- 2160 \rho^3 \sigma_1^2 - 4320 \rho^3 \sigma_0 \sigma_2 - 
   1440 \eta_0 \sigma_0 \sqrt{V_0} - 1440 \eta_1 \sigma_0 \sqrt{V_0} \nonumber \\
   & \qquad\qquad\qquad\qquad\qquad- 
   468 \eta_0 \rho^2 \sigma_0 \sqrt{V_0} - 6156 \eta_1 \rho^2 \sigma_0 \sqrt{V_0}- 
   2376 \eta_0 \rho^2 \sigma_1 \sqrt{V_0} \nonumber
   \\
   &\qquad\qquad\quad\qquad\qquad\qquad- 2736 \eta_1 \rho^2 \sigma_1 \sqrt{V_0} + 
   1440 \eta_0^2 \rho V_0 - 3816 \eta_0 \eta_1 \rho V_0 
   \nonumber
   \\
   &\qquad\qquad\quad\qquad\qquad\qquad\qquad+ 864 \eta_1^2 \rho V_0 - 
   12960 \eta_0 \eta_2 \rho V_0 \Big)\,,\nonumber
\end{align}
%%%%%%%%%%%%%%%%%%%%%%%%%
\begin{align}
\bar c_4 &=
\frac{1}{80 \eta_0^3 V_0^{ 3/2}}
\sigma_0 \Big(135 \rho \sigma_0^2 + 675 \rho^3 \sigma_0^2 - 2610 \rho \sigma_0 \sigma_1 + 
   2700 \rho^3 \sigma_0 \sigma_1 + 1350 \rho^3 \sigma_1^2  \\
   &\qquad\qquad\qquad\qquad+ 2700 \rho^3 \sigma_0 \sigma_2+ 
   360 \eta_0 \sigma_0 \sqrt{V_0} + 360 \eta_1 \sigma_0 \sqrt{V_0} + 
   495 \eta_0 \rho^2 \sigma_0 \sqrt{V_0}  \nonumber \\
   &\qquad\qquad\qquad\qquad\qquad+ 2295 \eta_1 \rho^2 \sigma_0 \sqrt{V_0}+ 
   1350 \eta_0 \rho^2 \sigma_1 \sqrt{V_0} + 1350 \eta_1 \rho^2 \sigma_1 \sqrt{V_0} 
   \nonumber
   \\
   &\qquad\qquad\qquad\qquad\qquad\qquad- 
   360 \eta_0^2 \rho V_0 + 1710 \eta_0 \eta_1 \rho V_0 - 360 \eta_1^2 \rho V_0 + 
   4860 \eta_0 \eta_2 \rho V_0 \Big)\,,\nonumber
\end{align}
%%%%%%%%%%%%%%%%%%%%%%%%%%%%%
\begin{align}
\bar c_5 &=
\frac{1}{80 \eta_0^3 V_0^{ 3/2}}
\sigma_0 \Big(45 \rho \sigma_0^2 - 135 \rho^3 \sigma_0^2 + 450 \rho \sigma_0 \sigma_1 - 
   540 \rho^3 \sigma_0 \sigma_1 - 270 \rho^3 \sigma_1^2  \\
   &\qquad\qquad\qquad\qquad- 540 \rho^3 \sigma_0 \sigma_2- 
   135 \eta_0 \rho^2 \sigma_0 \sqrt{V_0} - 315 \eta_1 \rho^2 \sigma_0 \sqrt{V_0} - 
   270 \eta_0 \rho^2 \sigma_1 \sqrt{V_0} 
   \nonumber
   \\
   &\qquad\qquad\qquad\qquad\qquad- 270 \eta_1 \rho^2 \sigma_1 \sqrt{V_0}- 
   270 \eta_0 \eta_1 \rho V_0 - 540 \eta_0 \eta_2 \rho V_0 \Big)\,.\nonumber
\end{align}

\end{document}